\documentclass[11 pt, onecolumn]{article}
\usepackage[margin=1.1in]{geometry}

\usepackage{subfigure}
\usepackage{amsmath,amssymb}

\usepackage{amsthm} 
\allowdisplaybreaks
\usepackage{color}
\usepackage{footnote}
\usepackage{algorithm}
\usepackage{algpseudocode}
\usepackage{algorithmicx}
\usepackage{multirow} 
\usepackage{booktabs}
\usepackage{graphicx}
\usepackage{amssymb}
\usepackage{amsbsy}
\usepackage{array}
\usepackage{longtable}
\usepackage{epstopdf}
\usepackage{pbox}
\usepackage{url}
\usepackage{breqn}
\usepackage{mathrsfs}
\usepackage{multicol}
\usepackage{supertabular}
\usepackage{enumerate}
\usepackage{bbm}
\usepackage{multirow}
\usepackage{hyperref}
\usepackage{cite}
\usepackage{xfrac}
\usepackage{textgreek}
\usepackage[justification=centering]{caption}
\usepackage{mathtools}

\usepackage{todonotes}
\usepackage{dsfont}
\newtheorem{dfn}{Definition}
\newtheorem{thm}{Theorem}

\newtheorem{prp}{Proposition}
\newtheorem{rem}{Remark}
\newtheorem{lmm}{Lemma}
\newtheorem{asm}{Assumption}

\newtheorem{pro}{Problem}
\newcommand{\argmin}{\arg\!\min}
\newcommand{\argmax}{\arg\!\max}

\algnewcommand\algorithmicforeach{\textbf{for each}}
\algdef{S}[FOR]{ForEach}[1]{\algorithmicforeach\ #1\ \algorithmicdo}
\usepackage{algorithm}
\usepackage{algpseudocode}
\usepackage{caption}

\DeclareMathAlphabet{\mathcal}{OMS}{cmsy}{m}{n}

\usepackage{graphicx}
\usepackage{epstopdf}

\begin{document}
\title{\huge \bf Optimizing Coordinated Vehicle Platooning: An Analytical Approach Based on Stochastic Dynamic Programming}

\author{Xi Xiong\footnote{Corresponding author: xi.xiong@nyu.edu.} \thanks{Department of Civil and Urban Engineering, New York University, Brooklyn, NY, USA.} ,
Junyi Sha\thanks{Tandon School of Engineering and Courant Institute of Mathematical Sciences, New York University, Brooklyn, NY, USA.},
and Li Jin$^\dagger$}
\newcommand*{\QEDA}{\hfill\ensuremath{\blacksquare}}%

\maketitle

\begin{abstract}
Platooning connected and autonomous vehicles (CAVs) can improve traffic and fuel efficiency.
However, scalable platooning operations require junction-level coordination, which has not been well studied.
In this paper, we study the coordination of vehicle platooning at highway junctions.
We consider a setting where CAVs randomly arrive at a highway junction according to a general renewal process.
When a CAV approaches the junction, a system operator determines whether the CAV will merge into the platoon ahead according to the positions and speeds of the CAV and the platoon.
We formulate a Markov decision process to minimize the discounted cumulative travel cost, i.e. fuel consumption plus travel delay, over an infinite time horizon.
We show that the optimal policy is threshold-based: the CAV will merge with the platoon if and only if the difference between the CAV's and the platoon's predicted times of arrival at the junction is less than a constant threshold.
We also propose two ready-to-implement algorithms to derive the optimal policy. Comparison with the classical value iteration algorithm implies that our approach explicitly incorporating the characteristics of the optimal policy is significantly more efficient in terms of computation. Importantly, we show that the optimal policy under Poisson arrivals can be obtained by solving a system of integral equations.
We also validate our results in simulation with Real-time Strategy (RTS) using real traffic data.
The simulation results indicate that the proposed method yields better performance compared with the conventional method.
\end{abstract}


\section{Introduction}

In the recent decade, the technology of connected and autonomous vehicles (CAVs) has been developing fast due to continuous progress in deploying communication and computation capabilities on vehicles and on infrastructure \cite{maurer2016autonomous, boccardi2014five}. Platooning is a novel highway operation enabled by the CAV technology, where vehicles travel in groups with very short inter-vehicle spacing \cite{davila2010sartre}. The major advantages of platooning include throughput improvement, fuel savings, and reduced pollutant emissions \cite{bhoopalam2018planning, tsugawa2011automated}. Such advantages lead to considerable incentives for the transportation industry to adopt this operation.
The architecture of a comprehensive platoon management system consists of three layers: the network layer, the link (junction) layer, and the vehicle layer \cite{alam2015heavy}.
The network layer deals with trip scheduling, goods assignment and route planning. The link layer coordinates the formulation, splitting, and reordering of platoons. 
The vehicle layer regulates the longitudinal and lateral motion of the vehicle in a platoon.
However, although much progress has been made for microscopic regulation of vehicle strings \cite{ploeg14,coogan15interconnected,besselink2017string,gao2017data}, very limited methods and results are available for link- and network-layer coordination for platooning.

In this paper, we consider a novel Markov decision process (MDP) formulation to study the coordinated platooning problem at independent highway junctions and propose an easy-to-implement but provably optimal strategy to coordinate CAVs. Figure \ref{fig:color_figure} illustrates the scenario that we consider. The local coordinator is located at the junction, and there are two flows of CAVs entering the junction. The coordinator has access to kinematic information (speeds and locations) of CAVs within the \emph{coordinating zone} of radius $D_1$ on each highway branch. The decision of whether a CAV merges into a platoon is made when the CAV arrives at the detector and enters the coordinating zone. If the coordinator decides to merge a CAV into a platoon, the CAV will be instructed to traverse the coordinating zone with a specified speed such that it can catch up with the nearest platoon ahead at the junction. There is a \emph{cruising zone} beyond the junction where CAV platoons that have formed in the coordinating zone will be maintained.
To design efficient coordination algorithms, we focus on the trade-off between the travel cost (time and fuel) over the coordinating zone due to acceleration/deceleration for platooning and the reduced fuel consumption over the cruising zone due to platooning.

\begin{figure}[hbt]
  \centering
  \includegraphics[width=0.7\textwidth, trim=100 150 100 120,clip]{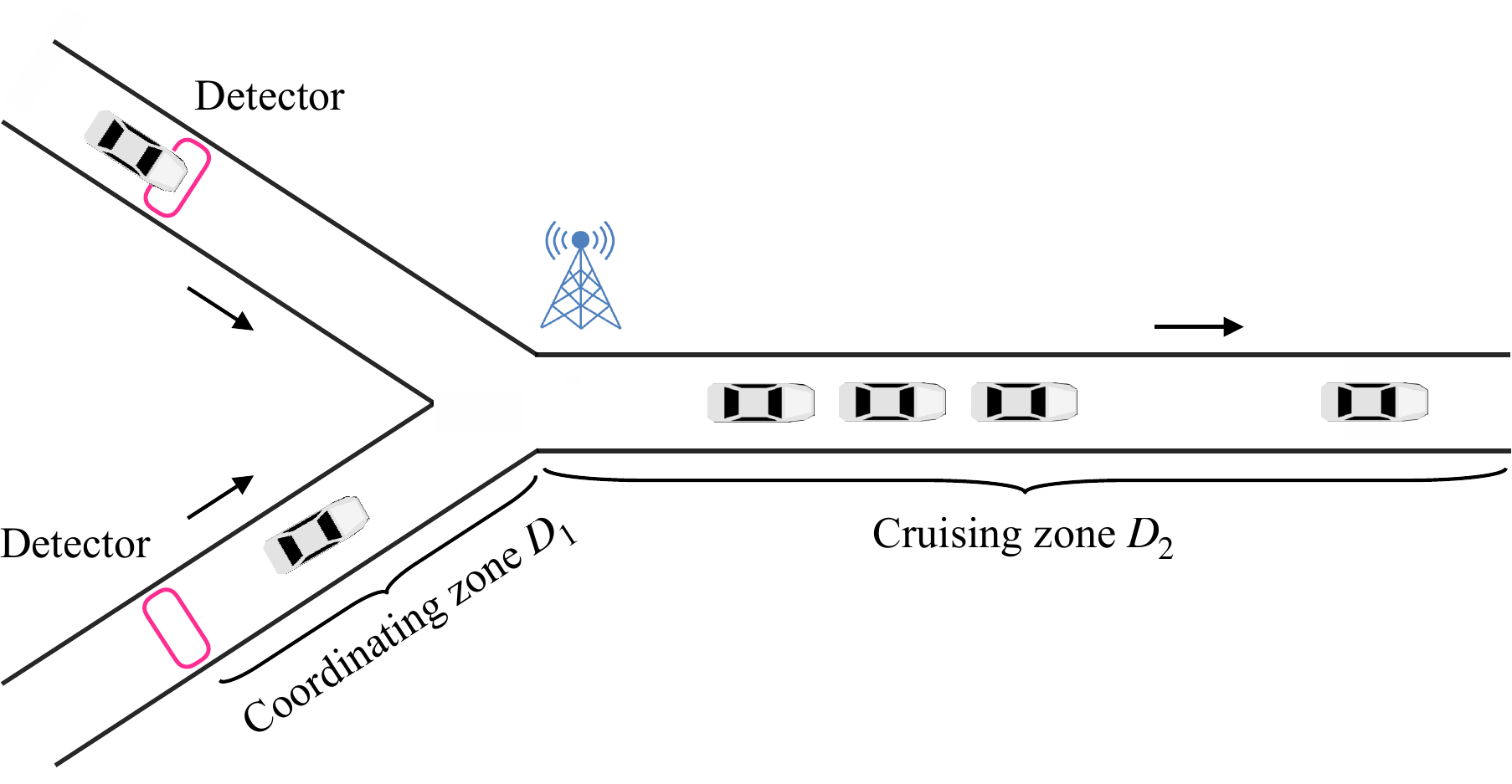}
  \caption{Coordinated platooning on a highway section.}\label{fig:color_figure}
\end{figure}

Most previous work on coordinated platooning is based on optimization (typically mixed integer programming) formulations with heuristic solution algorithms; see \cite{bhoopalam2018planning} for a rather complete overview.
In particular, Larson et al. \cite{larson2015distributed} considered a distributed formulation where platooning decisions are made at each junction; this distributed framework is practical and thus adopted in our modeling approach.
Larsson et al. \cite{larsson2015vehicle} proposed an integer programming for platooning-oriented routing and compared the performance of multiple heuristics.
Boysen et al. \cite{boysen2018identical} studied the coordinated platooning problem on a single route, which is similar to the setting considered in this paper, with various objective functions.
Luo et al. \cite{luo2018coordinated} used an integer programming to design coordination algorithms with multiple speeds.
Johansson et al. \cite{johansson2019game} considered a game-theoretic formulation to study the share of economical benefits due to platooning.
Sun and Yin \cite{sun2019behaviorally} quantifies the share of benefits due to platooning and studied the allocation of benefits.
In addition, there exists a line of work on higher-level trip planning for platooning \cite{deng2017energy,larsen2019hub,abdolmaleki2019itinerary}, which focuses on strategical routing and schedule.
There also exists a line of work on lower-level coordination between CAVs for platooning \cite{liang2015heavy,turri2016cooperative,van2017fuel,guo2018fuel}, which focuses more on the kinematics (and sometimes even dynamics) of individual vehicles rather than the overall traffic.
However, the above approaches require exact knowledge of each CAV's trip plan (or constraints), which is not always available. Furthermore, although heuristic algorithms are usually efficient and practical, it is also beneficial to derive structural and insightful results on optimal platooning strategies.
This paper is, to the best of our knowledge, the first effort that explicitly characterize the structure of optimal coordinated platooning strategies in an MDP environment.

In our formulation, CAVs enter the coordinating zone upstream to a highway junction as a general renewal process, i.e. with independent and identically distributed inter-arrival times. Note that if the inter-arrival times are exponentially distributed, then we have a Poisson process. 
The control action is the spatial-average speed in the coordinating zone recommended to each CAV, which determines whether a CAV will merge with the CAV(s) ahead to form a platoon. The objective is to minimize the travel cost, which is the sum of fuel consumption and travel delay. This objective function captures the trade-off between the fuel benefits due to platooning and the cost for forming platoons. The merging can be achieved by deceleration of leading vehicle or acceleration of following vehicle. Our formulation does not explicitly consider routing \cite{larson2015distributed} and travel time window constraints \cite{zhang2017freight, boysen2018identical}; the uncertainty/variation of each CAV's trip before entering the coordinating zone is captured by the stochastic arrival process. A similar model was considered in our previous work \cite{xiong2019analysis}, which studied a rather intuitive off-line, open-loop coordination strategy in a static setting. In this paper, we consider a MDP with feedback coordination strategies and use the open-loop strategy in \cite{xiong2019analysis} as a baseline for evaluation.

The main result that we derive for the MDP is an explicit characterization of the structure of an optimal policy (Theorem~\ref{thm_general}).
Such characterization is essential for computing the actual optimal policy, since standard dynamic programming (DP) algorithms do not directly apply to our formulation with a continuous and unbounded state space.
By analyzing the Bellman optimality equation for the MDP, we show that an optimal policy is \emph{threshold-based}: the following vehicle is supposed to merge with the leading vehicle if the predicted time interval between the arrivals of these two vehicles at the junction is less than a given threshold. This threshold essentially captures the comparison between the benefit (fuel savings over the cruising zone) and the cost (travel delay and/or additional fuel consumption over the coordinating zone) of platooning. If the predicted time interval is larger than the threshold, the following vehicle is supposed to travel with an optimal speed that is typically lower than the nominal speed, in anticipation of the arrival of the next vehicle as an opportunity for platooning.
The proof of this result uses an analytical approach based on stochastic dynamic programming. We first consider an $N$-stage, finite time horizon problem and show that the optimal policy for that problem is threshold-based. Then, we extend the argument to the infinite time horizon.
Our result is rather general in that (i) the arrival process is a general renewal process and (ii) the cost function is general (with assumptions on concavity).

We also develop ready-to-implement algorithms for computing the optimal policy in practical settings.
For general arrival processes, we compare two solution algorithms, viz. bounded value iteration (BVI) and recursive approximation (RA). BVI is derived from value iteration (VI), a classical approach in dynamic programming \cite{bertsekas1995dynamic}. Since the state space of our MDP is continuous and unbounded, the BVI first trims and discretizes the state space and then proceeds analogously to classical VI. Since BVI exhaustively iterates the value as well as the optimal action for each discretized state, it is computationally costly. The RA algorithm that we consider is intended to address this challenge by incorporating the threshold-based structure that we identified. Instead of computing the optimal actions for all states, the RA algorithm assumes the threshold-based structure a priori and only computes the parameters thereof; this significantly accelerates the computation.
Furthermore, for Poisson arrival processes, we show that the parameters of the optimal policy can be obtained by solving a system of integral equations (Theorem~\ref{thm_poisson}), which is much faster than iteration-based algorithms such as BVI and RA. 
The computational efficiency of these algorithms are compared in a numerical example with various arrival processes. 
In particular, for Poisson processes, the convergence times for the BVI, RA, and integral equation-based methods are 9.6 hours, 54 minutes, and less than 1 second.

We also validate our results in a realistic simulation environment. The simulation model is calibrated with real traffic data for a junction of US Interstate 210 (I210) obtained from the Freeway Performance Measurement System (PeMS \cite{varaiya2007freeway}).
We consider a hypothetical scenario where various percentages of all traffic are CAVs that can be platooned.
Simulation results show that, under our Real-time Strategy (RTS) of coordination, the average monetary savings (time plus fuel) at one junction is \$3736.8 per day.
In addition, we analyze the sensitivity of the optimal savings with respect to key model parameters, including the sizes of the cruising zone, and the discount factor for the MDP.
Finally, we present the interpretation of the RTS.

The rest of this paper is organized as follows.
In Section~\ref{sec_model}, we introduce the stochastic model and formulate the Markov decision process for coordinated platooning.
In Section~\ref{sec_analysis}, we characterize the structure of the optimal policy to the MDP.
In Section~\ref{sec_optimize}, we formulate and solve an optimization problem for coordinated platooning.
In Section~\ref{sec_results}, we present the numerical results using real traffic data.
In Section~\ref{sec_conclude}, we summarize the conclusions and propose several directions for future work. 
\section{Modeling and Formulation}
\label{sec_model}
In this section, we introduce our model for the process of coordinated platooning and formulate a Markov decision process (MDP) to minimize the system-wide travel cost based on this model.
Consider the scenario in Figure \ref{fig:color_figure}, where flows on two branches merge at the junction. Two detectors located equivalent distance $D_1$ from the junction on each branch would transmit the arrival time of each vehicle to the coordinator, which would determine whether the vehicle will merge with the previous platoon. When the coordinator sends the merging signal, the vehicle would accommodate its speed to meet the previous platoon at the junction, after which the vehicle will cruise together with the previous platoon, and experience fuel savings due to decreased air resistance. The vehicle will cruise with an optimal speed if the coordinator sends the non-merging signal. We assume that all vehicles share the identical path during the cruising zone $D_2$ to give insights into the one-junction coordination. The rest of this section is devoted to the details of our modeling and formulation.

\subsection{Modeling}
In this subsection, we define the arrival process of CAVs and the control actions that the platooning coordinator can take.
We also introduce the travel cost model for the platooning process, based on which we formulate the decision problem in the next subsection.

\subsubsection{Platooning coordination}

The coordination strategy is determined based on the arrival times on both branches. The following vehicle would merge with the previous platoon either on the same branch or on the other one. To simplify the modeling, two flows are assumed to appear on the same route, and one detector is used to record the arrival times of each vehicle. In Figure \ref{fig:SK}, the $(k-1)$th vehicle, $k=1,2, \ldots,$ enters the coordinating zone at time $T_{k-1}$. The following vehicle $k$ is recorded by either detector, and the arrival time is $T_k$. The inter-arrival time $X_k = T_k - T_{k-1}$ is assumed to follow an independent and identically distributed (i.i.d.) process, and the probability density function (PDF) is $f(x)$. The platooning strategy is realized by adapting the vehicle speed during the coordinating zone $D_1$. In real implementation, the intra-platoon headway $h_0$ is much smaller than coordinating distance $D_1$ and cruising distance $D_2$, then $h_0 \approx 0$ is assumed to simplify the analysis.

\subsubsection{State: predicted headway}
\label{Section: Predicted Headway}

Without loss of generality, we assume that vehicles cruise with an average speed $v$, hence the average transverse time on $D_1$ is $t_0 = \frac{D_1}{v}$. The $k$th vehicle would arrive at the junction at $(T_k + t_0)$ without changing the speed. In practice, vehicle $k$ can accelerate or decelerate to transverse the coordinating zone. We neglect the speed variation and assume that the vehicle would keep constant speed during the coordinating zone. In Figure \ref{fig:SK}, the actual arrival time would be earlier than $(T_k + t_0)$ when vehicle $k$ choose to accelerate, in which time reduction $U_k > 0$. Also $U_k < 0$ occurs in the deceleration case. The maximum time reduction for vehicle $k$ arises when vehicle $k$ meets vehicle $(k-1)$ at the junction. Then they would form a platoon, and vehicle $k$ would experience the reduced fuel consumption. We use $ \widetilde{S}_k =  \max \{U_k\}$ to denote the maximum time reduction for vehicle $k$. $\widetilde{S}_k$ can be negative when the actual arrival time of vehicle $(k-1)$ is larger than $(T_k + t_0)$, in which vehicle $k$ and vehicle $(k-1)$ would both drive with lower speeds. Note that the speed adaptation only occurs during the coordinating zone, and the spatial-average speed is assigned once passing the detector. After the junction, all vehicles would return to the average speed $v$. The coordinating zone is only used to coordinate platooning between consecutive vehicles by changing time reduction $U_k$ taking value from $\left(-\infty, \widetilde{S}_k \right]$, i.e., the speed during the coordinating zone, $v_k = \frac{D_1}{D_1 /v - U_k}$. Obviously, $\widetilde{S}_k$ should be less than $t_0$ if the two vehicles would merge at the junction.

A related work \cite{xiong2019analysis} designed the threshold-based policy based on the inter-arrival time $X_k$, which would result in frequent acceleration and thus high speed when forming long platoons. In addition, deceleration was not considered for the vehicle maneuver. In this paper, we incorporate the deceleration option, and use the time reduction $\widetilde{S}_k$ as the decision variable, which incorporates the information of the previous vehicle, and hence $\widetilde{S}_k$ is sensitive to platoon size. Specifically, the previous vehicle may leave the coordinating zone while the following vehicle arrives at the detection. In Figure \ref{fig:SK}, the actual arrival time at the junction for vehicle $k$ is $(T_k + t_0 - \widetilde{S_k})$ if it would merge with vehicle $(k-1)$. When vehicle $(k+1)$ enters the coordinating zone at $T_{k+1} > T_k + t_0 - \widetilde{S}_k$, vehicle $(k+1)$ cannot meet with vehicle $k$ at the junction. To generalize the decision variable for all vehicles, the concept of time reduction $\widetilde{S}_k$ is extended to the predicted headway $S_k$, which denotes the hypothetical time used to catch up with the first vehicle in the previous platoon. Consider the scenario where all vehicles drive towards the junction before the detector in Figure \ref{fig:platoon_series}. The predicted headway of vehicle $k$ still meets the equation $S_k = \widetilde{S}_k$. However, vehicle $(k+1)$ requires $(S_k + T_{k+1} - T_k)$ to catch up with vehicle $k$ (also vehicle $(k-1)$ due to $h_0 \approx 0$) even though vehicle $(k+1)$ would not merge with the previous platoon. Hence the predicted headway of vehicle $(k+1)$ is $S_{k+1} = S_k + T_{k+1} - T_k$. 

\begin{figure}
  \centering
  \includegraphics[width=0.75\textwidth, trim=200 240 200 230,clip]{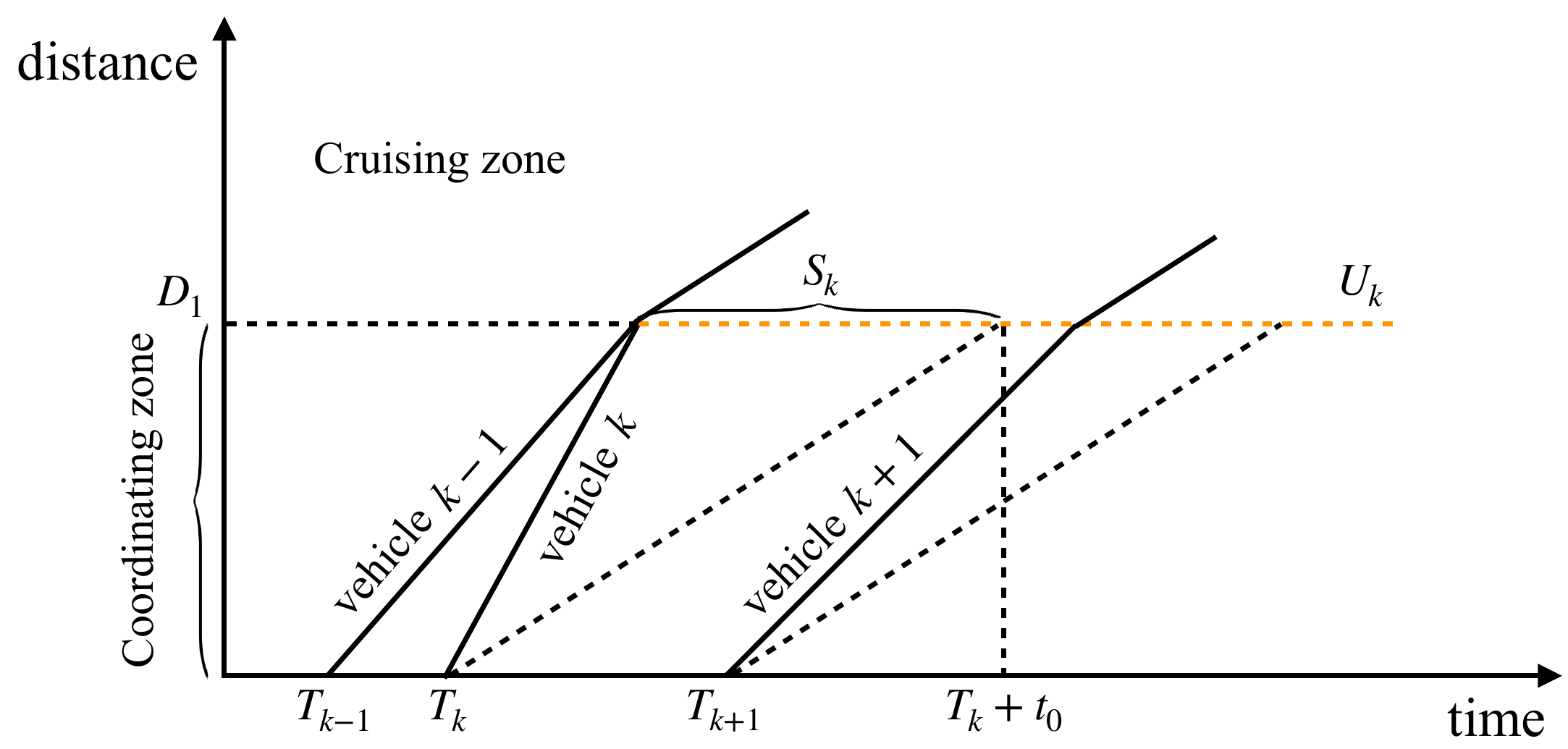}
  \caption{Predicted headway $S_k$ and time reduction $U_k$.}\label{fig:SK}
\end{figure}

\subsubsection{Problem definition}
We consider travel cost that consists of fuel consumption and travel time. When the vehicle is assigned to catch up with previous vehicle, the increased (decreased) speed would incur more (less) fuel consumption. The increased (decreased) speed during the coordinating zone would reduce (increase) the travel time. In addition, the vehicle would have less fuel consumption in $D_2$. We then formulate the optimization problem to minimize the travel cost. In section \ref{Section: Predicted Headway}, the predicted headway $S_k$ is proposed to be the decision variable for vehicle $k$. The action we used is the time reduction $U_k$, which affects the spatial-average speed, lies in $\left(-\infty , S_k \right]$ since the following vehicle could not surpass the leading vehicle. When $S_k \geq t_0$, the accessible time reduction $U_{k} \in \left(-\infty , t_0 \right)$.

The benefits of platooning consist of reduced travel time and improved fuel economy. Typically, we use the relative benefits instead of absolute values \cite{larson2015distributed}. The total cost for $k$th vehicle can be defined as:
\begin{align*}
    TC_k = -w_1 U_k + w_2 (\Delta F_1 - \Delta F_2),
\end{align*}
where $w_1$ represents the value of time, and $w_2$ denotes the fuel price. $\Delta F_1$ is the increased fuel during $D_1$, and $\Delta F_2$ is the fuel reduction resulting from reduced air resistance in $D_2$, which only arises when $k$th vehicle merges with $(k+1)$th vehicle.

Following the definition in \cite{larson2015distributed}, $\Delta F_1$ results from the air drag depending on the vehicle speed $v$ and the distance of $D_1$, $\Delta F_1 = \alpha D_1 \left( \frac{D_1}{D_1/v - U_k} \right)^2 - \alpha D_1 v^2$, where $\alpha$ represents the coefficient of increased fuel based on speed and distance. Typically, $\Delta F_2$ occurs when the following vehicle merges with the previous platoon, i.e., $U_k = S_k$, and is denoted as: $\Delta F_2 = \eta \phi D_2$, where $\eta$ represents the fuel saving fraction compared with driving alone, and $\phi$ denotes the fuel efficiency \cite{larson2015distributed}. When the following vehicle would not merge with the previous platoon, i.e., $U_k < S_k$, $\Delta F_2 = 0$. The total platooning cost for vehicle $k$ is:
\begin{align}
\label{Equation: total_cost}
    TC_k = \begin{cases}
    -w_1 U_k + w_2 \left(\alpha D_1 \left( \frac{D_1}{D_1/v - U_k} \right)^2 - \alpha D_1 v^2 - \eta \phi D_2 \right) & U_k = S_k, \\
    -w_1 U_k + w_2 \left(\alpha D_1 \left( \frac{D_1}{D_1/v - U_k} \right)^2 - \alpha D_1 v^2 \right) & U_k < S_k.
     \end{cases}
\end{align}
Note that $TC_k = -w_1 U_k + w_2 \left(\alpha D_1 \left( \frac{D_1}{D_1/v - U_k} \right)^2 - \alpha D_1 v^2 \right)$ when $S_k \geq t_0$, which also satisfies the case when $S_k \geq t_0 > U_k$.

To minimize the total travel cost for all vehicles, we need to find the optimal strategy of platooning based on the predicted headway. The optimization problem can be formulated as:
\begin{pro}
\label{Problem: Total Cost Minimization}
Consider $N$ vehicles, $N \to \infty$, drive towards the junction. The inter-arrival time $X_k$ is independent and identically distributed, and the probability density function is $f(x)$. The predicted headway for vehicle $k$ is $S_k$, $k = 1, 2, \ldots, N$. We need to find the optimal policy $\pi^* = \left\{\mu_1^*(s_1), \mu_2^*(s_2), \ldots \right\}$, to minimize the expected total travel cost for all vehicles.
\end{pro}

\subsection{MDP for coordinated platooning}
Specifically, the number of platooning combinations under $N$ vehicles is $2^N$, i.e., each vehicle can choose to merge with the previous vehicle or not (Figure \ref{fig:platoon_series}). In addition, the time reduction $U_k$ is continuous, which makes it difficult to find the $\pi^*$ in Problem \ref{Problem: Total Cost Minimization}. Since the inter-arrival time $X_k$ follows an independent and identically distributed (i.i.d.) process with the probability density function (PDF) $f(x)$, we then use the stochastic dynamic programming (SDP) \cite{bertsekas1995dynamic}, which captures the recursive interaction among consecutive vehicles, to find the optimal solution.

We formulate the problem as a discrete-time problem with $N$ control steps. The state at each step is the predicted headway $S_k$. The action $A_k$ is the time reduction $U_k$. The reward $R(s_k, a_k) = - TC_k$,
\begin{align}
    \label{Equation: Reward}
    R(s_k, a_k) = \begin{cases}
    w_1 a_k + w_2 \left(\alpha D_1 v^2 - \alpha D_1 \left( \frac{D_1}{D_1/v - a_k} \right)^2 + \eta \phi D_2 \right) & \ a_k = s_k, \\
    w_1 a_k + w_2 \left(\alpha D_1 v^2 - \alpha D_1 \left( \frac{D_1}{D_1/v - a_k} \right)^2 \right) & \ a_k < s_k, \\
     \end{cases}
\end{align}
where $a_k \in \left(-\infty, t_0 \right)$ when $s_k \geq t_0$, and $a_k \in \left(-\infty, s_k \right]$ when $s_k < t_0$. The reward function is discontinuous at $a_k=s_k$ when $s_k < t_0$ due to the fuel savings in $D_2$.

\begin{dfn}
We use $G(s)$ to denote the reward when $a = s$ with $s<t_0$,
\begin{align*}
   G(s) = w_1 s + w_2 \left(\alpha D_1 v^2 - \alpha D_1 \left( \frac{D_1}{D_1/v - s} \right)^2 + \eta \phi D_2 \right),
\end{align*}
and use $H(a)$ to denote the reward when $a<s$,
\begin{align*}
   H(a) = w_1 a + w_2 \left(\alpha D_1 v^2 - \alpha D_1 \left( \frac{D_1}{D_1/v - a} \right)^2 \right).
\end{align*}
Note that $G(s) - H(s) = G(0) = \eta \phi D_2 > 0$.
\end{dfn}

The derivative of $G(s)$ is,
\begin{align*}
    \frac{\mathrm{d} G}{\mathrm{d}s} =  w_1 - 2 w_2 \alpha \left(\frac{D_1}{D_1/v - s} \right)^3.
\end{align*}

Since $w_1$, $w_2$, and $\alpha$ are all positive, $G(s)$ is a concave function with the maximum value when $s=D_1 \left( \frac{1}{v} - \sqrt[3]{\frac{2 w_2 \alpha}{w_1}}\right)$ by letting $\frac{\mathrm{d} G}{\mathrm{d}s} = 0$. 
We suppose $c_N = \argmax_{s} {G(s)}$, then $c_N < \frac{D_1}{v} = t_0$. Similarly, $H(s)$ is concave and has the maximum value when $s = c_N$. In addition, $\lim\limits_{s \to t_0^- }{G(s)} = - \infty$, and $\lim\limits_{s \to - \infty }{G(s)} = - \infty$.

\begin{asm}
\label{Assumption 1}
    To generalize the total cost which balances fuel consumption and travel time, the generic total cost $G(s)$ follows:
    \begin{itemize}
        \item[1.] G(s) is a concave function for $s <t_0$, $\lim\limits_{s \to t_0^- }{G(s)} = - \infty$, and $\lim\limits_{s \to - \infty }{G(s)} = - \infty$;
        \item[2.] $c_N = \argmax_{s} G(s)$;
        \item[3.] For all $s$, $H(s) = G(s) - G(0)$, where $G(0)$ is the fuel savings of platooning in $D_2$.
    \end{itemize}{}
\end{asm}

The next state $S_{k+1}$ is the predicted headway for vehicle $(k+1)$. The deterministic relationship between vehicle $k$ and vehicle $(k+1)$ is,
\begin{align*}
    S_{k+1} = X_{k+1} + A_k.
\end{align*}

\begin{dfn}

We use $V_k^*(s)$, $k = 1,2, \ldots, N$, to denote the maximum expected total rewards, i.e., optimal state-value function, from step $k$ to step $N$, and then for all $k$ and $s \in \mathcal{S}$,
\begin{align*}
    V_k^*(s) \coloneqq \max_{\pi} {\mathbb{E} \left[{\sum_{j=0}^{N-k}}  \gamma^{j} R(s_{j+k}, a_{j+k}) | S_{k} = s \right]},
\end{align*}
where $0 < \gamma < 1$ is the discount factor, which denotes the extent of valuing long-term rewards. We use the discounted reward $\gamma^{j} R(s_{j+k}, a_{j+k})$ at each step since we start with vehicle $k$ and end with vehicle $N$. Let $Q_k^*(s,a)$ be the optimal action-value function,
\begin{align*}
    Q_k^*(s, a) \coloneqq \max_{\pi} {\mathbb{E} \left[{\sum_{j=0}^{N-k}}  \gamma^{j} R(s_{j+k}, a_{j+k}) | S_{k} = s, A_{k} = a, \pi \right]}.
\end{align*}
Then for all $k$, $s \in \mathcal{S}$ and $a \in \mathcal{A}$, 
\begin{align*}
V_k^*(s) = \max_{a \in \mathcal{A}} Q_k^*(s, a).
\end{align*}

\end{dfn}

The objective is to maximize the discounted total rewards $\mathscr{R}$,
\begin{align*}
\mathscr{R} = \lim\limits_{N \to \infty }{ {\mathbb{E} \left[{\sum_{j=0}^{N-1}}  \gamma^{j} R(s_{j+1}, a_{j+1}) \right]}},
\end{align*}
where $\mathscr{R}$ is associated with a policy $\pi = \left\{\mu_1, \mu_2, \ldots \right\}$

\begin{figure}
  \centering
  \includegraphics[width=1.0\textwidth, trim=90 470 60 190,clip]{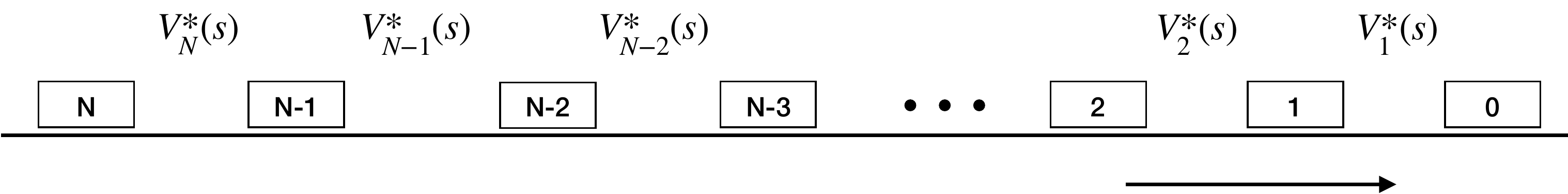}
  \caption{N-stage arrival of connected vehicles.}\label{fig:platoon_series}
\end{figure}

\section{Optimal Policy under General Arrival Processes}
\label{sec_analysis}

In this section, we study the structure of the optimal policy to the DP problem formulated in the previous section with general arrival processes for CAVs. Specifically, we show that the optimal strategy is threshold-based as follows:

\begin{thm}
\label{thm_general}
An optimal policy to the coordinated platooning Problem \ref{Problem: Total Cost Minimization} is a threshold-based policy such that
\begin{align}
\label{Equation: threshold_policy}
    \mu^*(s)=\begin{cases}
    s & s \leq \theta,\\
    c & s > \theta,
    \end{cases}
\end{align}
where $\theta$ is the threshold, and $c$ is the constant time reduction.
\end{thm}

One can interpret the optimal policy as follows.
When the $k$th vehicle enters the coordinating zone, if the predicted headway $s_k$ is less than or equal to a threshold $\theta$, then vehicle $k$ will be instructed to catch up with vehicle $(k-1)$ and thus platoon. Otherwise, vehicle $k$ will slightly slow down with the constant time reduction $c$ in anticipation of platooning with vehicle $(k+1)$.

We prove Theorem~\ref{thm_general} by analyzing the value function associated with the MDP.
We firstly study the finite-horizon, $N$-stage problem and characterize the structure of the optimal policy for $s \geq c_N$.
Then we take the limit of $N$ and obtain the result. Furthermore, the general threshold-based policy is presented. Throughout the derivation of this result, we assume a general arrival process and the generic cost function in Assumption \ref{Assumption 1}.

\subsection{Optimal returns of $N$-stage problem}
\label{Section: N-stage Problem}

\begin{prp}
\label{Proposition 1}
The optimal policy for vehicle $k$, $k = 1,2, \ldots, N$ is a threshold-based policy for $s \geq c_N$ such that
 \begin{align*}
    \mu^*_{k}(s_{k})=\begin{cases}
    s_{k} & c_N \leq s_{k} \leq \theta_{k},\\
    c_{k} & s_{k} > \theta_{k},
    \end{cases}
\end{align*}
where $\theta_k$ is the threshold for vehicle $k$, and $c_k$ is the optimal constant time reduction for vehicle k.
\end{prp}

We prove Proposition \ref{Proposition 1} by starting from the initial optimal returns for vehicle $N$, and then derive the optimal policies for all vehicles using the induction of value functions with known transition probabilities.

\begin{dfn}
    For all $s_k$, $a_k$ takes value from $\left(- \infty, s_k \right]$. We use $a_k^1 < s_k$ to denote the action of non-merging, and use $a_k^2 = s_k$ to represent the action of platooning with the previous vehicle.
\end{dfn}

\subsubsection{Initial optimal returns}
\begin{lmm}
\label{Lemma 1}
The optimal policy for vehicle $N$ is a threshold-based policy such that
\begin{align}
\label{Equation: policy_vehicle_N}
    \mu_N^*(s_N)=\begin{cases}
    s_N & s_N \leq \theta_N,\\
    c_N & s_N > \theta_N,
    \end{cases}
\end{align}
where $\theta_N $ is the threshold, $\theta_N \in (c_N, t_0)$, and $c_N$ is the optimal constant time reduction.
\end{lmm}{}

\begin{proof}
Since there is no vehicle after $N$th vehicle, $V_N^*(s) = \max_{a_N}R(s_N, a_N)$, for all $s_N$. 
We use $\left(Z_N + G(0) \right)$ to denote the maximum value (Figure \ref{fig:V_N(s)}), and $Z_N$ is the maximum value of $H(s)$ when $s<t_0$ due to $H(s) = G(s) - G(0)$. $\lim\limits_{s \to t_0^- }{G(s)} = - \infty$, hence
\begin{align*}
    \exists! \theta_N \in (c_N, t_0) \colon      \ G(\theta_N)=Z_N.
\end{align*}{}
For $s < c_N$, $\lim\limits_{s \to - \infty }{G(s)} = - \infty$, hence
\begin{align*}
    \exists! \theta_N^' \in (-\infty, c_N) \colon \ G(\theta_N^')=Z_N.
\end{align*}{}

\begin{figure}[hbt]
  \centering
  \includegraphics[width=0.43\textwidth,trim=250 250 220 210,clip]{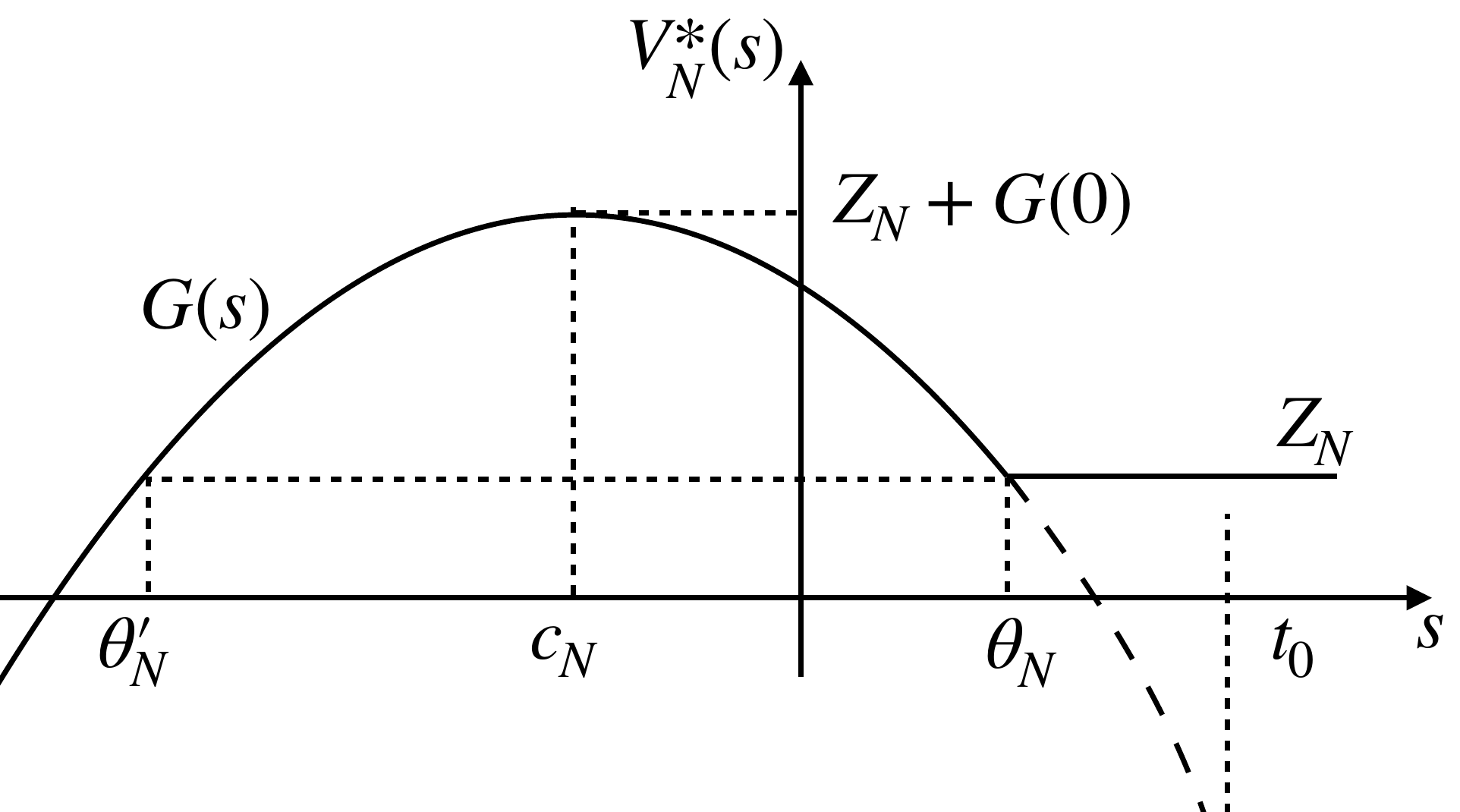}
  \caption{Value function for vehicle $N$.}
  \label{fig:V_N(s)}
\end{figure}

We then analyze the action-value functions under different $s$:
\begin{itemize}

\item[(i)] When $s_N \in \left( -\infty,  c_N \right]$, for all  $a_N^1 < s_N$ and $a_N^2 = s_N$,
\begin{align*}
    Q_N^*(s_N,a_N^1) & = H(a_N^1) < H(s_N), \\
    Q_N^*(s_N,a_N^2) & = G(s_N),
\end{align*}
and then $Q_N^*(s_N,a_N^1) < H(s_N) < G(s_N) = Q_N^*(s_N,a_N^2)$. The optimal action is $a_N^2=s_N$, i.e., merging with the previous vehicle.

\item[(ii)] When $ s_N \in \left( c_N,  \theta_N \right]$, for all  $a_N^1 < s_N$ and $a_N^2 = s_N$,
\begin{align*}
    Q_N^*(s_N,a_N^1) & = H(a_N^1) < H(c_N)=Z_N, \\
    Q_N^*(s_N,a_N^2) & = G(s_N),
\end{align*}
and then $Q_N^*(s_N,a_N^1) < Z_N \leq G(s_N) = Q_N^*(s_N,a_N^2)$. The optimal action is $a_N^2=s_N$.

\item[(iii)] When $ s_N \in \left( \theta_N,  t_0 \right)$, for all  $a_N^1 < s_N$ and $a_N^2 = s_N$,
\begin{align*}
    & c_N  = \argmax_{a_N^1} Q_N^*(s_N,a_N^1), \\
    & Q_N^*(s_N, c_N)  = Z_N, \\
    & Q_N^*(s_N,a_N^2) = G(s_N) < Z_N,
\end{align*}
and then the optimal action is $a_N^1 = c_N$, i.e., cruising with the non-merging time reduction $c_N$.

\item[(iv)] When $s_N \in \left[t_0, +\infty \right)$,
\begin{align*}
    & Q_N^*(s_N,a_N) = H(a_N), \\
    & c_N = \argmax_{a_{N}} Q_N^*(s_N,a_N),
\end{align*}
and then the optimal action is also $a_N = c_N$.
\end{itemize}

Above all, the optimal policy for vehicle $N$ is a threshold-based policy as shown in Lemma \ref{Lemma 1}.
Furthermore, the value function $V_N^*(s)$ is:
\begin{align*}
   V_N^*(s) = \begin{cases}
    G(s) & s \leq \theta_N, \\
    Z_N & s > \theta_N.
    \end{cases}
\end{align*}

\end{proof}

\begin{rem}
Vehicle $N$ should merge with the previous vehicle if the predicted headway is below the threshold $\theta_N$, otherwise it should cruise with the optimal constant time reduction $c_N$. In practice, we usually value fuel consumption over travel time, i.e., $w_2 \gg w_1$, then $c_N=D_1 \left( \frac{1}{v} - \sqrt[3]{\frac{2 w_2 \alpha}{w_1}}\right) < 0 $. Vehicle $N$ would cruise with lower speed with time reduction $c_N$.

\end{rem}

\subsubsection{Induction of value function}
\label{Section: Induction for K vehicle}

\begin{lmm}
\label{Lemma 2}

If vehicle $k$, $k= 2,\ldots, N$, follows an optimal threshold-based policy for $s \geq c_N$ such that
 \begin{align*}
    \mu^*_{k}(s_{k})=\begin{cases}
    s_{k} & c_N \leq s_{k} \leq \theta_{k},\\
    c_{k} & s_{k} > \theta_{k},
    \end{cases}
\end{align*}
where $\theta_{k}$ is the threshold for vehicle $k$, $\theta_{k} \in \left(c_N, \theta_N \right] $, and $c_k$ is the optimal constant time reduction for vehicle $k$, $ c_k \in \left( \theta_N^', c_N \right]$. In addition, $V^*_k(s)$ is monotonically decreasing for $s \geq c_N$, $V^*_k(s) = Z_k$ for $s \geq \theta_k$, and $V^*_k(c_k) = Z_k + G(0)$.

Then vehicle $(k-1)$ also follows an optimal threshold-based policy for $s \geq c_N$,
\begin{align*}
    \mu^*_{k-1}(s_{k-1})=\begin{cases}
    s_{k-1} & c_N \leq s_{k-1} \leq \theta_{k-1},\\
    c_{k-1} & s_{k-1} > \theta_{k-1},
    \end{cases}
\end{align*}
where $\theta_{k-1}$ is the threshold for vehicle $(k-1)$, $\theta_{k-1} \in \left(c_N, \theta_N \right) $, and $c_{k-1}$ is the optimal constant time reduction for vehicle $(k-1)$, $ c_{k-1} \in \left( \theta_N^', c_N \right] $. $V^*_{k-1}(s)$ is monotonically decreasing for $s \geq c_N$, $V^*_{k-1}(s) = Z_{k-1}$ for $s \geq \theta_{k-1}$, and $V^*_{k-1}(c_{k-1}) = Z_{k-1} + G(0)$.

\end{lmm}

\begin{proof}
We consider vehicle $k$ and vehicle $(k-1)$,
\begin{align*}
V^*_{k-1}(s) = \max_{a_{k-1} \in \mathcal{A}_{k-1}} \left[ R(s_{k-1}, a_{k-1}) +  \gamma \sum_{s_k \in \mathcal{S}_k} \mathbb{P}(s_k | s_{k-1}, a_{k-1}) V^*_{k}(s_k) \right].
\end{align*}
When $ s_{k-1} < t_0$, for all $a_{k-1}^1 < s_{k-1}$ and $a_{k-1}^2 = s_{k-1}$,
\begin{align*}
    Q_{k-1}^*(s_{k-1},a_{k-1}^1) = H(a_{k-1}^1) + \gamma \int_0^{+\infty} f(x) V_k^*(a_{k-1}^1+x) \mathrm{d}x, \\
    Q_{k-1}^*(s_{k-1},a_{k-1}^2) = G(s_{k-1}) + \gamma \int_0^{+\infty} f(x) V_k^*(s_{k-1}+x) \mathrm{d}x.
\end{align*}

We use $J_{k-1}(s)$ to denote $Q_{k-1}^*(s_{k-1},a_{k-1}^2)$,
\begin{align}
\label{Equation: J_{k-1}}
    J_{k-1}(s) = G(s) + \gamma \int_0^{+\infty} f(x) V_k^*(s+x) \mathrm{d}x, \ \forall s < t_0,
\end{align}
and then $Q_{k-1}^*(s_{k-1},a_{k-1}^1) = J_{k-1}(a_{k-1}^1) - G(0)$ for all $a_{k-1}^1 < s_{k-1}$.

For $s \leq \theta_N^'$,
\begin{align*}
    J_{k-1}(s) = \underbrace{G(s)}_{\leq Z_N} + \underbrace{\gamma \int_0^{+\infty} f(x) V_k^*(s+x) \mathrm{d}x}_{< \gamma \left( Z_k + G(0) \right)},
\end{align*}
We then derive the lower bound of $J_{k-1}(c_{N})$,
\begin{align}
\label{Equation: lower_bound_J_{k-1_c_N}}
    J_{k-1}(c_{N}) = \underbrace{G(c_N)}_{= Z_N + G(0)} + \underbrace{\gamma \int_0^{+\infty} f(x) V_k^*(c_N + x) \mathrm{d}x}_{> \gamma Z_k},
\end{align}
hence
\begin{align*}
    J_{k-1}(s) < Z_N + \gamma \left( Z_k + G(0) \right) < Z_N + G(0) + \gamma Z_k < J_{k-1}(c_{N}), \ \forall s \leq \theta_N^'.
\end{align*}

Since $\frac{\mathrm{d} G}{\mathrm{d}s} < 0$, and $V_k^*(s)$ is monotonically decreasing for $ s \geq c_N$, $J_{k-1}(s)$ in Equation \eqref{Equation: J_{k-1}} is monotonically decreasing in $\left[c_N, t_0 \right)$ for all $f(x)$ (Figure \ref{fig:V_K_1(s)}), we suppose
\begin{align*}
 c_{k-1} = \sup \left\{s \in  (\theta_N^', c_N ] \colon J_{k-1}(s) = Z_{k-1} + G(0)   \right\},
\end{align*}
where $Z_{k-1} + G(0) = \max_{s \in \left(- \infty, c_N \right]} J_{k-1}(s)$.

\begin{figure}[hbt]
  \centering

  \subfigure[Value function of $k$th vehicle.]
  {\label{fig:V_K(s)}
  \includegraphics[width=0.42\textwidth, trim=250 230 250 210,clip]{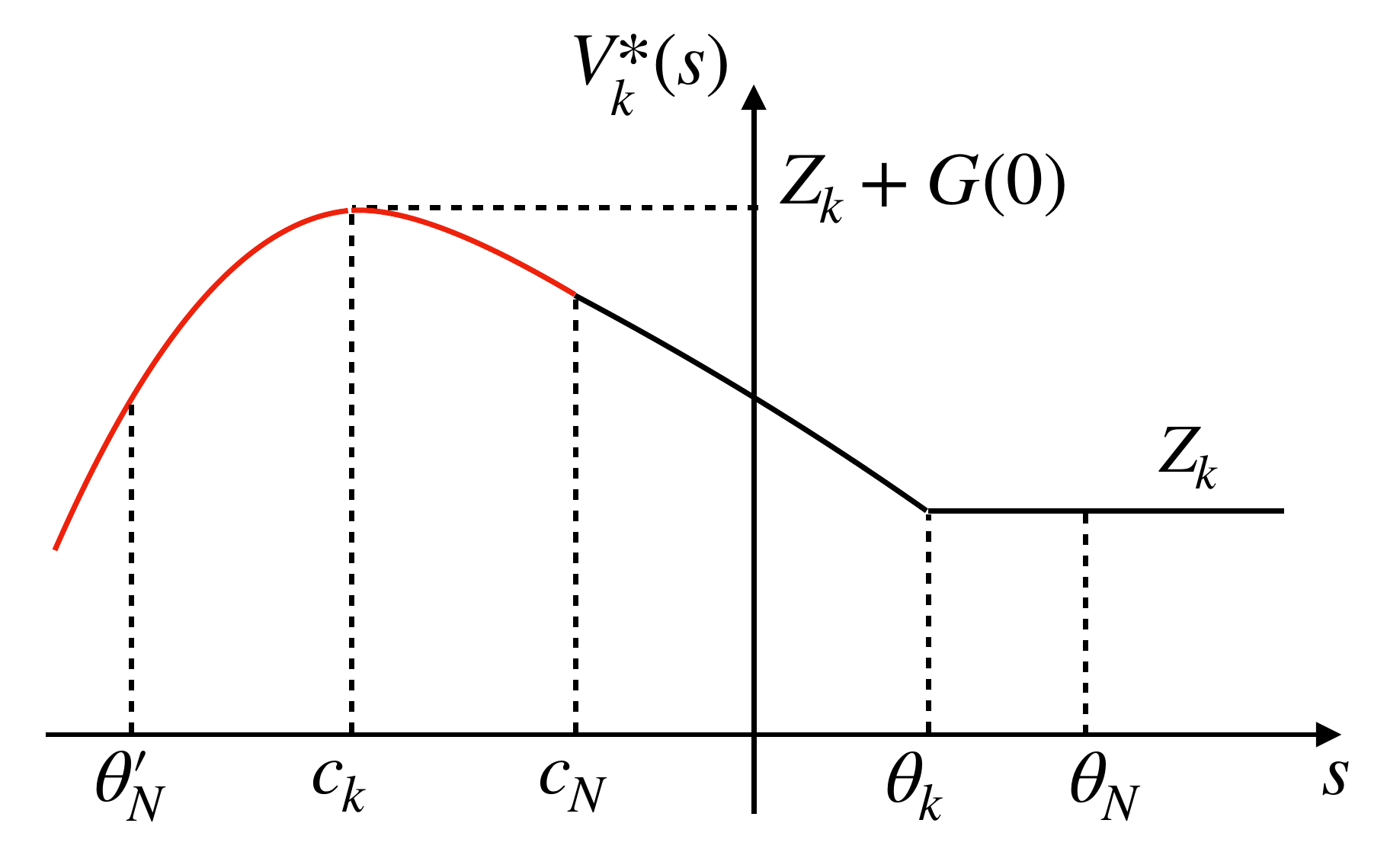}
  }
  \subfigure[Value function of $(k-1)$th vehicle.]
  {\label{fig:V_K_1(s)}
  \includegraphics[width=0.42\textwidth, trim=250 230 250 210,clip]{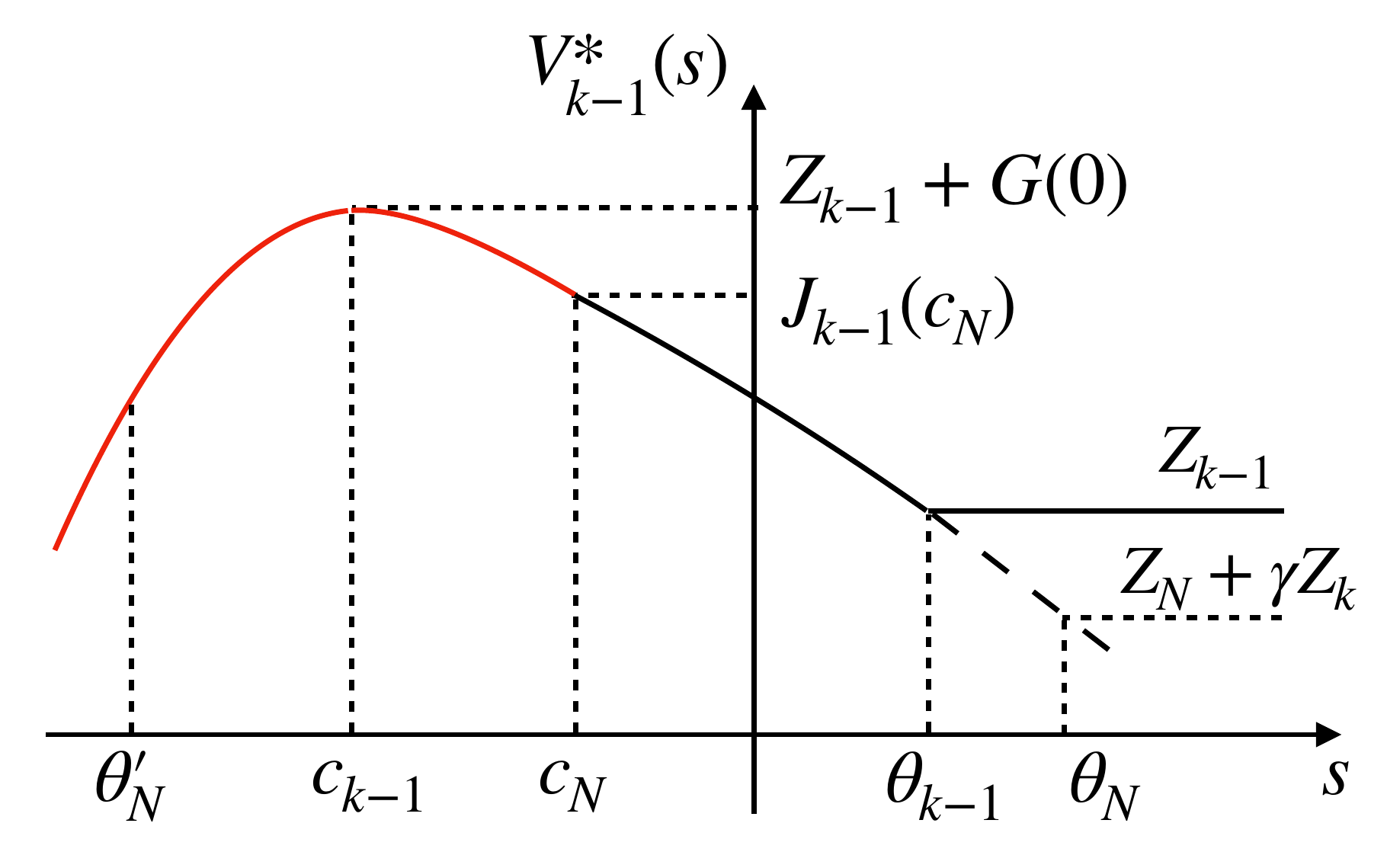}
  }
  
  \caption{Induction of value function for vehicle $k$ and $(k-1)$.}
\end{figure}

Then we analyze the upper bound and lower bound of $Z_{k-1}$.

\begin{itemize}
    \item[1.] Firstly we compare $Z_{k-1}$ with $J_{k-1}(\theta_N)$,
\begin{align*}
    Z_{k-1} - J_{k-1}(\theta_N) & = J_{k-1}(c_{k-1}) - G(0) - G(\theta_N) - \gamma \int_0^{+\infty} f(x) V_k^*(\theta_N + x) \mathrm{d}x \\
    & \geq J_{k-1}(c_{N}) - G(0) - Z_N - \gamma Z_k.
\end{align*}
In Equation \eqref{Equation: lower_bound_J_{k-1_c_N}}, $J_{k-1}(c_{N}) > Z_N + G(0) + \gamma Z_k$,
hence
\begin{align}
    \label{Equation: lower_bound_Z_k_1}
    Z_{k-1} - J_{k-1}(\theta_N) > 0.
\end{align}

    \item[2.] Then we compare $Z_{k-1}$ with $J_{k-1}(c_N)$,
\begin{align*}
    Z_{k-1} - J_{k-1}(c_N) & = J_{k-1}(c_{k-1}) - G(0) - G(c_N) - \gamma \int_0^{+\infty} f(x) V_k^*(c_N + x) \mathrm{d}x \\
    & = \underbrace{G(c_{k-1})}_{\leq G(c_{N})} + \underbrace{\gamma \int_0^{+\infty} f(x) V_k^*(c_{k-1} + x) \mathrm{d}x}_{\leq \gamma \left(Z_k + G(0) \right)} - G(0) - G(c_N) \\
    & \quad - \underbrace{\gamma \int_0^{+\infty} f(x) V_k^*(c_N + x) \mathrm{d}x}_{\geq \gamma Z_k},
\end{align*}
hence
\begin{align*}
    Z_{k-1} - J_{k-1}(c_N) \leq G(c_{N}) + \gamma \left( Z_k + G(0) \right) - G(0) - G(c_N) - \gamma Z_k < 0,
\end{align*}
and then
\begin{align}
    \label{Equation: upper_bound_Z_k_1}
    Z_{k-1} - J_{k-1}(c_N) < 0.
\end{align}

\end{itemize}

From Equation \eqref{Equation: lower_bound_Z_k_1} and Equation \eqref{Equation: upper_bound_Z_k_1}, $J_{k-1}(\theta_N)< Z_{k-1} < J_{k-1}(c_N)$. $J_{k-1}(s)$ is monotonically decreasing in $\left[c_N, t_0 \right)$, hence 
\begin{align*}
\exists! \theta_{k-1} \in (c_N, \theta_N) \colon \ J_{k-1}(\theta_{k-1})=Z_{k-1}.
\end{align*}

We then analyze the action-value functions under different $s$:
\begin{itemize}

\item[(i)] When $s_{k-1} \in \left[c_{N}, \theta_{k-1}
\right]$, for all  $a_{k-1}^1 < s_{k-1}$ and $a_{k-1}^2 = s_{k-1}$,
\begin{align*}
    Q_{k-1}^*(s_{k-1},a_{k-1}^1) & = J_{k-1}(a_{k-1}^1) - G(0) \leq Z_{k-1}, \\
    Q_{k-1}^*(s_{k-1},a_{k-1}^2) & \geq Z_{k-1},
\end{align*}
and then $Q_{k-1}^*(s_{k-1},a_{k-1}^1) \leq Q_{k-1}^*(s_{k-1},a_{k-1}^2)$. The optimal action is $a_{k-1}^2 = s_{k-1}$, i.e., merging with the previous vehicle.

\item[(ii)] When $s_{k-1} \in \left(\theta_{k-1}, t_0 \right]$, for all  $a_{k-1}^1 < s_{k-1}$ and $a_{k-1}^2 = s_{k-1}$,
\begin{align*}
    & c_{k-1} = \argmax_{a_{k-1}^1} Q_{k-1}^*(s_{k-1},a_{k-1}^1) \\
    & Q_{k-1}^*(s_{k-1},c_{k-1}) = Z_{k-1}, \\
    & Q_{k-1}^*(s_{k-1},a_{k-1}^2) < Z_{k-1},
\end{align*}
and then the optimal action is $a_{k-1}^1 = c_{k-1}$, i.e., cruising with the non-merging time reduction $c_{k-1}$.

\item[(iii)] When $ s_{k-1} \in \left[t_0, +\infty \right]$, $a_{k-1} < t_0$
\begin{align*}
    & c_{k-1} = \argmax_{a_{k-1}} Q_{k-1}^*(s_{k-1},a_{k-1}), \\
    & Q_{k-1}^*(s_{k-1},c_{k-1}) = Z_{k-1},
\end{align*}
and then the optimal action is $a_{k-1} = c_{k-1}$.

\end{itemize}

Above all, the optimal policy for vehicle $(k-1)$ is a threshold-based policy for $s \geq c_{N}$ as shown in Lemma \ref{Lemma 2}. Furthermore, the value function $V^*_{k-1}(s)$ is:
\begin{align*}
   V_{k-1}^*(s) = \begin{cases}
    J_{k-1}(s) & c_N \leq s \leq \theta_{k-1}, \\
    Z_{k-1} & s > \theta_{k-1}.
    \end{cases}
\end{align*}
Note that we cannot know $V_{k-1}^*(s)$ for $s < c_{N}$ due to the unknown $f(x)$. We will clarify this in Section \ref{Section: Threshold_policy}.

\end{proof}

Since vehicle $N$ follows the optimal policy shown in Lemma \ref{Lemma 1}, and $V^*_N(s)$ meets the conditions in Lemma \ref{Lemma 2}, vehicle $k$, $k = 1,2,\ldots, N-1$ also follows an optimal threshold-based policy for $s \geq c_N$ as shown in Proposition \ref{Proposition 1}, and $V^*_k(s)$ has the characteristics in Lemma \ref{Lemma 2}.

\subsection{Extension to infinite horizon}
\label{Section: Infinite Horizion}
\begin{prp}
\label{Proposition 2}
When $N \to \infty$, the optimal policy for vehicle $k$, $k = 1,2,\ldots$, converges to the threshold-based policy for $s \geq c_N$ such that
 \begin{align*}
    \mu^*_{}(s_{})=\begin{cases}
    s_{} & c_N \leq s_{} \leq \theta_{},\\
    c_{} & s_{} > \theta_{},
    \end{cases}
\end{align*}
where $\theta$ is the threshold, $\theta \in \left[c_N, \theta_N \right]$, and $c$ is the optimal time reduction, $c \in \left[ \theta_N^', c_N \right]$.
\end{prp}{}

\begin{proof}
Since we consider the infinite horizon discounted MDPs, there exists an optimal \emph{stationary policy} \cite{bertsekas1995dynamic}, $\pi^* = \left\{\mu^*, \mu^*, \mu^*, \ldots \right\}$. In Proposition \ref{Proposition 1}, we have shown that vehicle $k \in \left\{1,2,\ldots, N \right\}$ follows the policy with the threshold $\theta_k$ and time reduction $c_k$ for $s \geq c_N$. Then $\mu_k^*(s_k)$ should converge to $\mu^*(s)$ shown in Proposition \ref{Proposition 2} with $\theta \in \left[c_N, \theta_N \right]$ and $c \in \left[ \theta_N^', c_N \right]$.

When $N \to \infty$, $V_k^*(s)$ would converge to $V^*(s)$, $Z_k$ would converge to $Z$, and $V^*(s)$ is monotonically decreasing for $s \geq c_N$. In addition,
\begin{align}
\label{Equation: V^*(theta)_Z}
    V^*(s) = Z, \ \forall s \geq \theta,
\end{align}
and
\begin{align}
\label{Equation: V^*(c)}
    V^*(c) = Z + G(0).
\end{align}

The Bellman optimality equation for $V^*(s)$ is,
\begin{align*}
V^*(s) = \max_{a \in \mathcal{A}} \left[  R(s, a) +  \gamma \sum_{s' \in \mathcal{S}} \mathbb{P}(s' | s, a) V^*(s') \right], \ \forall s \in \mathcal{S},
\end{align*}
where $s'$ is the next state of $s$.

Since the probability density function is $f(x)$,
 \begin{align*}
    V^*(s) = \begin{cases}
    Z & s > \theta, \\
    G(s) + \gamma \int_0^{+\infty} f(x) V^*(s+x) \mathrm{d}x & c_N \leq s \leq \theta,\\
    \max_{a \in ( -\infty, s] } \left[  R(s, a) +  \gamma \int_0^{+\infty} f(x) V^*(a+x) \mathrm{d}x \right] & s < c_N,
    \end{cases}
\end{align*}
where $V^*(s)$ is continuous at $s = \theta$.

\end{proof}

\subsection{The general threshold-based policy}
\label{Section: Threshold_policy}

\begin{prp}
\label{Proposition 3}
    The optimal stationary policy for $s < c_N$ is,
     \begin{align*}
    \mu^*_{}(s_{}) = s,
\end{align*}
where the vehicle would merge with the previous vehicle.
\end{prp}

We first present the lower bound of the value function difference, and then compare the action-value functions to derive the policy for $s < c_N$.

\subsubsection{Lower bound of vaue function difference}
\label{Section: lower_bound_value_function_difference}

\begin{lmm}
\label{lemma 3}
    For $s_1$, $s_2 \in \mathcal{S}$, $s_2 > s_1$, then 
    \begin{align}
        V^*(s_2) - V^*(s_1) \geq - G(0).
    \end{align}
    
\end{lmm}

\begin{proof}

We suppose,
\begin{align*}
    V^*(s_1) =  \max_{a} Q^*(s_1, a) = Z' + G(0),
\end{align*}
where $Z' \leq Z$ (Figure \ref{fig:V(s)}).

Then we derive the lower bound of $\left( V^*(s_2) - V^*(s_1) \right)$ under the following scenarios:
\begin{itemize}
    \item [(i)] When $s_2 > s_1 \geq \theta$, $V^*(s_1) = V^*(s_2) = Z' + G(0)$,
    \begin{align*}
        V^*(s_2) - V^*(s_1) = 0 > - G(0).
    \end{align*}
    
    \item [(ii)] When $s_2 \geq \theta > s_1$,
    \begin{align*}
        V^*(s_2) - V^*(s_1) = Z - \left(Z' + G(0) \right) \geq - G(0). 
    \end{align*}
    
    \item [(iii)] 
    When $s_1 < s_2 \leq \theta$,
    \begin{itemize}
        \item[(a)] We suppose $a_{1}^1 = \argmax_a {Q^*(s_1, a)}$, where $a_{1}^1$ denotes the non-merging action for $s_1$,
        \begin{align*}
         Q^*(s_2, a_{1}^1) = H(a_{1}^1) + \gamma \int_0^{+\infty} f(x) V^*(a_{1}^1 + x) \mathrm{d}x = Q^*(s_1, a_{1}^1) = Z' + G(0),
        \end{align*}
        hence
        \begin{align*}
            V^*(s_2) = \max_{a \in ( -\infty, s_2] } Q^*(s_2, a) \geq Z' + G(0),
        \end{align*}
        and then
        \begin{align*}
            V^*(s_2) - V^*(s_1) \geq 0 > - G(0).
        \end{align*}
        
        \item[(b)] $a_{1}^2 = \argmax_a {Q^*(s_1, a)}$, where $a_{1}^2$ denotes the merging action for $s_1$,
        \begin{align*}
         Q^*(s_1, a_1^2) = G(s_1) + \gamma \int_0^{+\infty} f(x) V^*(s_1 + x) \mathrm{d}x = Z' + G(0).
        \end{align*}
        Since $s_1 < s_2$,
        \begin{align*}
         Q^*(s_2, s_1) = H(s_1) + \gamma \int_0^{+\infty} f(x) V^*(s_1 + x) \mathrm{d}x = Z',
        \end{align*}
        hence 
        \begin{align*}
            V^*(s_2) = \max_{a \in ( -\infty, s_2] } Q^*(s_2, a) \geq Z',
        \end{align*}
        and then
        \begin{align*}
            V^*(s_2) - V^*(s_1) \geq - G(0).
        \end{align*}
        
    \end{itemize}
\end{itemize}

Above all,
\begin{align*}
    V^*(s_2) - V^*(s_1) \geq - G(0), \ \forall s_2 > s_1.
\end{align*}

\begin{figure}[hbt]
  \centering
  \includegraphics[width=0.61\textwidth, trim=150 250 220 210,clip]{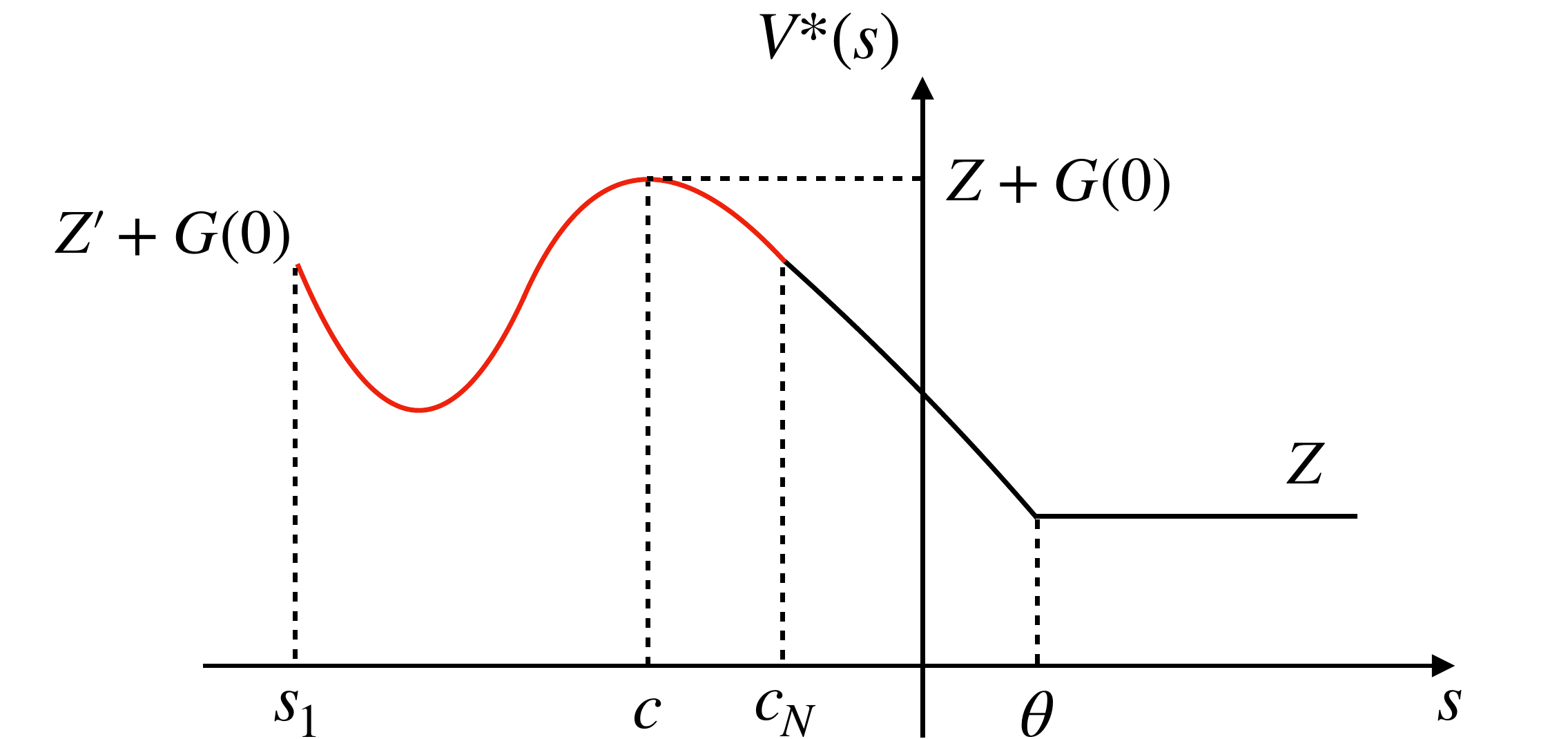}
  \caption{Optimal value function $V^*(s)$.} \label{fig:V(s)}
\end{figure}

\end{proof}

\subsubsection{Action value comparison}
\label{Section: value_function_comparison}

\begin{dfn}
For $s \in \mathcal{S}$, the action $a \in \left(-\infty, s \right]$. $a_1^* < s$ is the action of non-merging, and $a_2^* = s$ is the action of merging. Note that when $s \geq t_0$, $a \in \left(-\infty, t_0 \right)$.
\end{dfn}

\begin{lmm}
\label{lemma 4}
    For $s < c_N$, 
    \begin{align}
        Q^*(s, a_2^*) > Q^*(s, a_1^*),
    \end{align}
    where $a_1^* < s$, and $a_2^* = s$. The optimal action for $s < c_N$ is to merge with the previous vehicle.
\end{lmm}

\begin{proof}
    For $s < c_N$,
\begin{align*}
    V^*(s) = \max_{a \in ( -\infty, s] } Q^*(s, a),
\end{align*}
Then for the action $a_1^* < a_2^* = s$,
\begin{align*}
    Q^*(s, a_1^*) = H(a_1^*) + \gamma \int_0^{+\infty} f(x) V^*(a_1^* + x) \mathrm{d}x, \\
    Q^*(s, a_2^*) = G(s) + \gamma \int_0^{+\infty} f(x) V^*(s + x) \mathrm{d}x.
\end{align*}
We take the difference,
\begin{align*}
    Q^*(s, a_2^*) - Q^*(s, a_1^*) = G(s) - H(a_1^*) + \gamma \int_0^{+\infty} f(x) \left[ V^*(s + x) - V^*(a_1^* + x) \right] \mathrm{d}x.
\end{align*}

Since $s+x > a_1^* + x$, the lower bound of the state-value function difference (Lemma \ref{lemma 3}) is
\begin{align*}
     V^*(s + x) - V^*(a_1^* + x) \geq -G(0), \ \forall x > 0,
\end{align*}
and then for all $f(x)$,
\begin{align*}
    Q^*(s, a_2^*) - Q^*(s, a_1^*) \geq G(s) - H(a_1^*) - \gamma G(0).
\end{align*}
For $a_1^* < s <c_N$,
\begin{align*}
    G(s) - H(a_1^*) - \gamma G(0) > G(0) - \gamma G(0) > 0,
\end{align*}
hence
\begin{align*}
    Q^*(s, a_2^*) > Q^*(s, a_1^*), \ \forall s < c_N.
\end{align*}

Then the optimal policy for $s < c_N$ is to merge with the previous vehicle as shown in Proposition \ref{Proposition 3}.

\end{proof}{}

In summary, from Proposition \ref{Proposition 2} and Proposition \ref{Proposition 3}, the optimal stationary policy is,
\begin{align*}
    \mu^*(s)=\begin{cases}
    s & s \leq \theta,\\
    c & s > \theta,
    \end{cases}
\end{align*}
where $\theta$ is the threshold, $\theta \in \left[c_N, \theta_N \right]$, and $c$ is the optimal time reduction, $c \in \left[\theta_N^', c_N \right]$.
The optimal value function is,
 \begin{align}
 \label{Equation: final_V^*(s)}
    V^*(s) = \begin{cases}
    Z & s > \theta, \\
    G(s) + \gamma \int_0^{+\infty} f(x) V^*(s+x) \mathrm{d}x & s \leq \theta.
    \end{cases}
\end{align}


\section{Solution Algorithm}
\label{sec_optimize}

In this section, we propose algorithms to compute the optimal coordinated platooning policies.
For general arrival processes, we design two dynamic programming (DP) algorithms based on previous analysis: bounded value iteration and recursive approximation, to determine the optimal $\theta$ and $c$ in Equation \eqref{Equation: threshold_policy}.
The bounded value iteration is a generic DP algorithm with the threshold-structure (Section~\ref{sub_bvi}), while the recursive approximation (Section~\ref{sub_ra}) explicitly incorporates the characteristics of the optimal policy obtained in the previous section.
We compare the computational efficiency of these two algorithms in various scenarios (Section~\ref{sub_algorithm_compare}).
For Poisson arrival processes (which is the most commonly used model for vehicle arrivals), we further show that $\theta$ and $c$ can be obtained by solving a system of integral equations (Section~\ref{sub_poisson}), which is significantly faster than DP-based algorithms.

\subsection{Bounded value iteration}
\label{sub_bvi}
Value iteration is a generic approach in dynamic programming \cite{sutton2018reinforcement}. According to our analysis in Section \ref{Section: Threshold_policy}, the value function $V^*(s)$ stays constant when $s > \theta_N$. We only need to update the value function for $s \leq \theta_N$ to reduce algorithm complexity. The bounded value iteration is shown in Algorithm \ref{algorithm:value_iteration}.

\begin{algorithm}[htbp]
  \caption{Bounded value iteration.}
  \label{algorithm:value_iteration}
  \begin{algorithmic}[1]
    \Require
      Discrete states $s \in \left[m,n \right]$;
      Reward function $R(s, a)$;
      General headway probability density function  $f(x)$;
      Initial value function $V^*(s)$;
      Discount factor $\gamma$;
      Threshold of vehicle $N$ $\theta_N$;
      Small positive number $\varepsilon$ determining the estimation accuracy;
    \Ensure
      Optimal platooning policy $\pi^*(s)$;
    \While {$\Delta \geq \varepsilon$}
        \State $\Delta \gets 0$;
    \ForEach {$s \in \left[m, n \right]$}
        \State $v \gets V^*(s)$;
        \If{ $s \leq \theta_N$}
            \State $V^*(s) \gets \max_a \left[R(s,a) + \gamma \int_0^{+\infty} f(x) V^*(a + x) \mathrm{d}x \right]$
        \Else
             \State $V^*(s) \gets v$
        \EndIf
        \State $\Delta \gets \max(\Delta, |v - V^*(s)|)$
    \EndFor
    \EndWhile
\State Output the optimal deterministic policy $\pi^*$ such that
\State $\pi^*(s) = \argmax_a \left[R(s,a) + \gamma \int_0^{+\infty} f(x) V^*(a + x) \mathrm{d}x \right]$
\end{algorithmic}
\end{algorithm}

The action is the time reduction $a$, which lies in $\left(-\infty, s \right]$. Although the state $s$ and action $a$ are both continuous, the problem is transformed into the tabular framework due to the discontinuous reward $R(s, a)$ at $a = s$ shown in Equation \eqref{Equation: Reward}. In Algorithm \ref{algorithm:value_iteration}, we firstly discretize the states in a given range $\left[m,n \right]$, $m < c_N < \theta_N < n$. Then we let $V^*(s) \big|_{s > \theta_N} = V^*(s) \big|_{s = \theta_N}$ to simplify the value iteration process. The optimal policy would converge when the estimation gap is smaller than the small positive number $\varepsilon$, which determines the estimation accuracy.

\subsection{Recursive approximation}
\label{sub_ra}
The complexity of value iteration in Algorithm \ref{algorithm:value_iteration} is determined by the small positive number $\varepsilon$. The iteration steps would increase significantly if we choose a small $\varepsilon$. To improve the efficiency, we design the algorithm (Algorithm \ref{algorithm:solve_equation}) that approximates the optimal solution based on the characteristics of $V^*(s)$.

In Equation \eqref{Equation: final_V^*(s)}, $V^*(s)$ is continuous at $ s = \theta$,
\begin{align*}
    V^*(\theta) & = G(\theta) + \gamma \int_0^{+\infty} f(x) V^*(\theta + x) \mathrm{d}x = G(\theta) + \gamma Z = Z.
\end{align*}
hence,
\begin{align}
\label{Equation: Poisson_2}
Z = \frac{G(\theta)}{1-\gamma}.
\end{align}
For each discrete $\theta_i \in \left[m,n \right]$, we can derive the $Z_i$ based on Equation \eqref{Equation: Poisson_2}. Since the optimal solution is a threshold-based policy, $V_i^*(s) = Z_i$ for $s \geq \theta_i$. Then we can derive $V_i^*(s)$ for $s < \theta_i$ recursively from $\theta_i$ to $m$ using the following equation,
\begin{align*}
V_i^*(s) = G(s) + \gamma \int_0^{+\infty} f(x) V_i^*(s + x) \mathrm{d}x.
\end{align*}
For each $\theta_i$, we can derive the specific function $V_i^*(s)$. In Section \ref{Section: Infinite Horizion}, $V_i^*(s)$ has the maximum value $M_i= Z_i + G(0)$ when $s = c_i$. The optimal $\theta$ and $c$ should be the values that make $M_i$ approximate $\left( Z_i + G(0) \right)$ best. The methodology of Algorithm \ref{algorithm:solve_equation}
is in accordance with the essence of dynamic programming since we derive $V_i^*(s)$ recursively from $\theta_i$ to $m$, which would avoid the value iteration.

\begin{algorithm}[htb]
  \caption{Recursive approximation.}
  \label{algorithm:solve_equation}
  \begin{algorithmic}[1]
    \Require
      Discrete states $s \in \left[m,n \right]$;
      Discrete threshold $\theta \in \left[c_N, \theta_N \right]$;
      Reward function $R(s, a)$;
      General headway probability density function  $f(x)$;
      Initial value function $V^*(s)$;
      Dicsount factor $\gamma$;
    \Ensure
      Optimal threshold $\theta$ and $c$;
      
    \State Inverse the list of states $s$;
    \ForEach{$\theta_i$ in $\left[c_N, \theta_N \right]$}
    
        \State $Z_i \gets \frac{G(\theta_i)}{1-\gamma}$ ;
        
        \ForEach{$s$ in the reversed list from $n$ to $m$}
        
            
            \If {$s \geq \theta_i$}
                \State $V_i^*(s) \gets Z_i$;
                
            \Else
            
                \State $V_i^*(s) \gets G(s) + \gamma \int_0^{+\infty} f(x) V_i^*(s + x) \mathrm{d}x$.
  
            \EndIf
        \EndFor
        
        \State $M_i \gets \max_s V_i^*(s)$;
        \State $c_i \gets \argmax_s V_i^*(s)$;
        
    \EndFor
    
\State Optimal $c, \theta$ = $\argmin_{c_i, \theta_i} \left| M_i - \left(G(0) + Z_i \right) \right|$
  \end{algorithmic}
\end{algorithm}

\subsection{Optimal strategy under Poisson arrivals}
\label{sub_poisson}

\begin{thm}
\label{thm_poisson}
The optimal policy to the coordinated platooning Problem \ref{Problem: Total Cost Minimization} under Poisson arrivals with $f(x) = \lambda e^{- \lambda x}$ is a threshold-based policy such that
\begin{align}
    \mu^*(s)=\begin{cases}
    s & s \leq \theta, \\
    c & s > \theta,
    \end{cases}
\end{align}
where $\theta$ and $c$ satisfy the following equations:
\begin{align}
\label{Equation: Poisson_solutions}
    \begin{cases}
        Z - e^{\lambda (1-\gamma) \theta} \left( \int_{c}^{\theta} e^{ - \lambda (1-\gamma) t} \left(G^'(t) - \lambda G(t) \right) \mathrm{d}t + \left(Z + G(0) \right) e^{- \lambda (1-\gamma) c} \right) = 0, \\
        G(\theta) + \gamma Z - Z = 0, \\
        G^'(c) -\lambda G(c) + \lambda (1 - \gamma) \left( Z + G(0) \right)= 0,
    \end{cases}
\end{align}
in which $\lambda$ is the arrival rate, $G'(s)$ is the derivative of $G(s)$.
\end{thm}

\begin{proof}
For all $s < \theta$,
\begin{align*}
    V^*(s) & = G(s) + \gamma \int_0^{+\infty}\lambda e^{- \lambda x} V^*(x+s) \mathrm{d}x \\
    & = G(s) + \gamma e^{\lambda s} \int_{s}^{+\infty}\lambda e^{- \lambda x} V^*(x) \mathrm{d}x,
\end{align*}
in which the integral interval $\left(s, +\infty \right)$ can be divided into $\left(s, \theta \right)$ and $\left(\theta, + \infty \right)$.
In Equation \eqref{Equation: V^*(theta)_Z}, $V^*(s) = Z $ for s $\geq \theta$,
\begin{align*}
    V^*(s) & = G(s) + \gamma e^{\lambda s} \int_{s}^{\theta}\lambda e^{- \lambda x} V^*(x) \mathrm{d}x  + \gamma e^{\lambda s}  \int_{\theta}^{+ \infty}\lambda e^{- \lambda x} V^*(x) \mathrm{d}x \\
          & = G(s) + \gamma e^{\lambda s} \int_{s}^{\theta}\lambda e^{- \lambda x} V^*(x) \mathrm{d}x  + \gamma e^{\lambda s} Z e^{- \lambda \theta}.
\end{align*}
The derivative of $V^*(s)$ is,
\begin{align}
\label{Equation: ODE}
    \frac{\mathrm{d} V^*(s)}{\mathrm{d} s} = G^'(s) -\lambda G(s) + \lambda (1 - \gamma) V^*(s),
\end{align}
which meets the format of Ordinary Differential Equation (ODE), the solution of the ODE \cite{chicone2006ordinary} with $V^*(c) = Z + G(0)$ in Equation \eqref{Equation: V^*(c)} is
\begin{align*}
    V^*(s) = e^{\lambda (1-\gamma) s} \left( \int_{c}^{s} e^{ - \lambda (1-\gamma) t} \left(G^'(t) - \lambda G(t) \right) \mathrm{d}t + \left(Z + G(0) \right) e^{- \lambda (1-\gamma) c} \right).
\end{align*}
Since $V^*(\theta) = Z$,
\begin{align}
\label{Equation: Poisson_1}
    Z = e^{\lambda (1-\gamma) \theta} \left( \int_{c}^{\theta} e^{ - \lambda (1-\gamma) t} \left(G^'(t) - \lambda G(t) \right) \mathrm{d}t + \left(Z + G(0) \right) e^{- \lambda (1-\gamma) c} \right).
\end{align}
$V^*(s)$ has the maximum value when $s=c$, and $V^*(s)$ is derivable for $s < \theta$,
\begin{align*}
    \frac{\mathrm{d} V^*}{\mathrm{d} s} \bigg|_{s = c} = G^'(c) -\lambda G(c) + \lambda (1 - \gamma) V^*(c) = 0,
\end{align*}
hence
\begin{align}
\label{Equation: Poisson_3}
G^'(c) -\lambda G(c) + \lambda (1 - \gamma) \left(Z + G(0) \right)= 0.
\end{align}
We can find the optimal threshold $\theta$ and constant time reduction $c$ by solving Equation \ref{Equation: Poisson_2}, \ref{Equation: Poisson_1}, and \ref{Equation: Poisson_3}.

\end{proof}

\subsection{Algorithm efficiency comparison}
\label{sub_algorithm_compare}
In this subsection, we compare the efficiency of the proposed algorithms: Bounded value iteration (BVI), Recursive approximation (RA), and the refinement under Poisson arrivals (PR). The time complexity comparison is shown in Table \ref{table:complexity_analysis}. $|\mathcal{S}|$ is the number of states in $\left[m,n \right]$, $|\mathcal{S}_1|$ is the number of states in $\left[m, \theta_N \right]$, $|\mathcal{S}_2|$ is the number of states in $\left(\theta_N, n \right]$, and $|\mathcal{S}_3|$ is the number of states in $\left[c_N, \theta_N \right]$, $\mathcal{S}_3 \subset \mathcal{S}_1 \subset \mathcal{S}$.
$|\mathcal{I}|$ is the number of iterations to converge in Algorithm \ref{algorithm:value_iteration}. Since $\mathcal{I}$ is large due to the small $\varepsilon$, the RA algorithm is more efficient than BVI under general arrival processes. The time of solving a system of integral equations is constant.

\begin{table}[htbp]
  \caption{Time complexity analysis.}\label{table:complexity_analysis}
  \begin{center}
    \begin{tabular}{cccc}
    \hline
    Algorithms              &  BVI     &RA      &PR   \\
    \hline
    Time complexity         & $\mathcal{O} \left( \left(|\mathcal{S}_1|^2 + |\mathcal{S}_2| \right) \times |\mathcal{I}| \right)$         & $\mathcal{O}(|\mathcal{S}| \times |\mathcal{S}_3|)$                         & $\mathcal{O}(1)$ \\
    \hline
    \end{tabular}
  \end{center}
\end{table}

Then we use nominal values to evaluate the algorithm efficiency.
$m=-100 \ s$, $n=400 \ s$, the discrete interval between consequent $s$ is $0.25 \ s$, $\theta_N = 27.5 \ s$, $c_N = -0.49 \ s$. In Algorithm \ref{algorithm:value_iteration}, the small positive number $\varepsilon$ is $0.002$. Other nominal values are shown in Table \ref{table:case_study_parameters}. Three typical headway distributions are chosen to evaluate the solution time: exponential, discrete random variable, and constant headway, where $\lambda = 0.02 \ veh/s$, $p_1 = 0.4$, $p_2 = 0.6$, $h_{c1} = 15.0 \ s$, $h_{c2} = 8.0 \ s$, and $h_c = 10.0 \ s$. The algorithms are evaluated on the platform with Xeon Bronze 3104 CPU $@$ 1.70GHz.

The time of finding the optimal solution is shown in Table \ref{table:numerical_results}. When the headway follows an exponential distribution, the solution time of bounded value iteration is $9.6$ hours, while the solution time of recursive approximation is less than $1$ hour. When the discrete random and constant headway distributions are chosen, recursive approximation  still yields better performance than bounded value iteration. In addition, the time in exponential distribution is more significant due to the calculation of the integral. Specifically, solving the system of equations is the most efficient approach under Poisson arrivals. The results indicate that adding more prior knowledge can improve algorithm efficiency.

\begin{table}[!ht]
  \caption{Comparison of the solution time.}\label{table:numerical_results}
  \begin{center}
    \begin{tabular}{ccccc}
    \hline
    Headway Distribution & $f(x)$ & BVI & RA & PR   \\
    \hline
    Exponential  & $\lambda e^{- \lambda x}$   &  $9.6$ h    & $0.9$ h & $<1$s                              \\
    
    Discrete random variable     & $p_1 h_{c1} + p_2 h_{c2}$     & $37$ s   &   $10$ s  & $ - $              \\
    
    Constant     & $h_c$      & $29$ s   &   $8$ s   & $ - $           \\
    \hline
    \end{tabular}
  \end{center}
\end{table}

\section{Numerical Results}
\label{sec_results}

In this section, we use real data to evaluate our proposed policy presented in the previous sections. The Real-time Strategy (RTS) is developed to coordinate platooning under different flows. The case study setting is provided in Section~\ref{sub_setting}. The validation is conducted in three steps. In Section~\ref{sub_performance}, we compare the threshold-based policy with a straightforward policy and quantify the improvement. In Section~\ref{sub_sensitivity}, we study the sensitivity of the average cost with respect to key operational parameters. In Section~\ref{sub_interpretation}, we present the interpretation of the RTS.

\subsection{Setting}
\label{sub_setting}
We study the implementation of heavy-duty vehicle platooning at the interchange of I-210 and 134 located in the Los Angles metropolitan area. The real data are acquired through the Freeway Performance Measurement System (PeMS), which provides historical traffic data based on Caltrans loop detectors \cite{varaiya2007freeway}. The 24-hour data on January 22, 2019 are used to evaluate the policy (Table \ref{table:vehicle_flow}). Poisson arrivals are assumed to generate vehicle departures during each hour, and the flow within each hour keeps constant.

Detectors are located $D_1 = 1 \ km$ from the junction on freeway I-210 and 134. The platooning decision is made based on the predicted headway introduced in Section~\ref{Section: Predicted Headway}. When the following vehicle is determined to merge with the leading vehicle, the following vehicle would adjust its speed to meet the leading vehicle at the junction. We consider a safety reaction time $2.3 \ s$ for the following vehicle to avoid collision \cite{mcgehee2000driver}, i.e., the following vehicle would arrive $2.3 \ s$ later than the expected arrival time given by our policy. Note that there is no safety reaction time under non-merging case. The maximum speed on the freeway in our setting is $40 \ m/s $.

\begin{table}[!ht]
  \caption{Vehicle flow on freeway I-210 and 134.}\label{table:vehicle_flow}
  \begin{center}
    \begin{tabular}{cccccc}
    \hline
    
    \multirow{2}*{Time interval}   & I-210 flow   & 134 flow  
    &\multirow{2}*{Time interval}  & I-210 flow   & 134 flow \\
    &[veh/hour] &[veh/hour] &      &[veh/hour]    &[veh/hour]\\
    \hline
    $0:00-1:00$    & 254         & 665       & $12:00-13:00$  & 1005        & 4212     \\
    
    $1:00-2:00$    & 249         & 525       & $13:00-14:00$  & 1074        & 4351     \\
    
    $2:00-3:00$    & 210         & 445       & $14:00-15:00$  & 1205        & 5252     \\
    
    $3:00-4:00$    & 206         & 407       & $15:00-16:00$  & 1251        & 5568     \\
    
    $4:00-5:00$    & 269         & 687       & $16:00-17:00$  & 1374        & 5711     \\                
    $5:00-6:00$    & 397         & 1398      & $17:00-18:00$  & 1340        & 5783     \\                  
    $6:00-7:00$    & 693         & 3164      & $18:00-19:00$  & 1351        & 5677     \\                  
    $7:00-8:00$    & 1221        & 5498      & $19:00-20:00$  & 1150        & 4989     \\                 
    $8:00-9:00$    & 1367        & 5740      & $20:00-21:00$  & 845         & 3696     \\                
    $9:00-10:00$   & 1141        & 4922      & $21:00-22:00$  & 745         & 2934     \\                
    $10:00-11:00$  & 1040        & 4286      & $22:00-23:00$  & 582         & 2233     \\                     
    $11:00-12:00$  & 946         & 3993      & $23:00-24:00$ & 413         & 1361      \\ 
    \hline
    \end{tabular}
  \end{center}
\end{table}


Since the vehicle flow fluctuates during the day, we consider the $M$-step discounted previous headways as an estimator for the arrival rate:
\begin{align*}
      \hat\lambda_k=\Big[(1-\beta)\sum_{m=0}^{M-1}\beta^mX_{k-m}\Big]^{-1},
\end{align*}
where $\beta \in \left(0, 1 \right)$ is the discounted factor of inter-arrival time, and larger $\beta$ denotes longer headway calculation period. $X_{k-m}$ is the $m$th previous headway. When $M$ is large enough, $\hat\lambda_k$ can be an estimator for the real arrival rate $\lambda$. 

The estimated arrival rate $\hat\lambda_k$ is updated once the vehicle is detected on either branch. Under Poisson arrivals, the optimal threshold and C could be obtained by solving a system of integral equations in Theorem~\ref{thm_poisson}. In our case study, Real-time Strategy (RTS) is realized by solving Equation \eqref{Equation: Poisson_solutions} based on dynamic prediction of the arrival rate, i.e., each vehicle would have the specific threshold and C value, which would yield better performance with respect to the flow fluctuation. The solution of Equation \eqref{Equation: Poisson_solutions} is calculated using the \textit{Python} package \textit{scipy.optimize}, which relies on the initial value. To avoid the divergence in real implementation, the calculated threshold and C value of vehicle $k$ are the initial values of vehicle $(k+1)$.

The evaluation metric is the absolute total cost including fuel consumption and travel time. Fuel consumption is calculated by the sum of fuel cost at every step. The fuel rate $f \ (L/s)$ is derived using the data in Simulation of Urban MObility (SUMO) \cite{krajzewicz2002sumo}: $f = 3.51 \times 10^{-7} v^3 + 4.07 \times 10^{-4} v$, in which we neglect the effect of acceleration. The fuel consumption in the coordinating zone is: $F_1 = \frac{D_1}{v} f = 3.51 \times 10^{-7} D_1 v^2 + 4.07 \times 10^{-4} D_1$, hence $\Delta F_1 = 3.51 \times 10^{-7} \left( D_1 v_k^2 - D_1 v^2 \right)$, in which $v_k$ is the modified speed of vehicle $k$, and then $\alpha = 3.51 \times 10^{-7} \ L \cdot s^2 / m^3$.
Fuel consumption and travel time both include cost in the coordinating zone $D_1$ and cruising zone $D_2$. Nominal parameters shown in Table \ref{table:case_study_parameters} are used in the case study.

\begin{table}[htbp]
  \caption{Nominal parameters in the case study.}\label{table:case_study_parameters}
  \begin{center}
    \begin{tabular}{cc}
    \hline
    Parameter             & Value     \\
    \hline
    $\eta$                & $0.1$       \\
    $\alpha$              & $3.51\times10^{-7} \ L \cdot s^2 / m^3$ \\
    $\phi$                & $32.2 \ L/100 km$       \\
    $v$                   & $23 \ m/s$      \\
    $w_1$                 & $25.8 \ \$/hour $       \\
    $w_2$                 & $0.868 \ \$ / L$       \\
    $D_1$                 & $1 \ km$       \\
    $D_2$                 & $30 \ km$       \\
    $\gamma$              & $0.9$       \\
    $\beta$               & $0.9$       \\
    $M$                   & $50$       \\
    \hline
    \end{tabular}
  \end{center}
\end{table}

\subsection{Performance comparison}
\label{sub_performance}
In this subsection, we compare three platooning policies under different flows:
\begin{enumerate}
    \item  Baseline: all vehicles keep the original speed during the coordinating zone, and the following vehicle would merge with the leading vehicle if they arrive at the junction within the safety reaction time $2.3 \ s$.
    \item Policy A: we use the proposed policy in \cite{xiong2019analysis}. The following vehicle would merge with the previous vehicle if their inter-arrival time is below an optimal threshold, otherwise it would keep the original speed $v$. Note that only acceleration is enabled in this policy.
    \item Policy B: our proposed policy is implemented according to Theorem~\ref{thm_general}.
\end{enumerate}
Note that Policy A and Policy B are both RTS due to the dynamic prediction of arrival rate.

Since we study the platooning of heavy-duty vehicles, the arrival rate is much smaller than the real traffic flow. We assume that a percentage of real flow on I-210 and 134 is the heavy-duty vehicle, and all vehicles in the simulation can be connected. Figure \ref{fig:average_total_cost} shows total cost under different average flows, i.e., different percentages of Average Daily Traffic (ADT) in Table \ref{table:vehicle_flow}. The evaluation metric is the average cost per vehicle, $AC = \frac{TC}{N_t}$, where $TC$ is the total cost, and $N_t$ is the total number of vehicles. The results show that our proposed policy yield the minimum average total cost compared with other policies. The improvement is more significant under higher average flows.

\begin{figure}[htbp]
  \centering
  \subfigure[Average total cost]
  {\label{fig:average_total_cost}
  \includegraphics[width=0.32\textwidth,trim=30 90 110 120,clip]{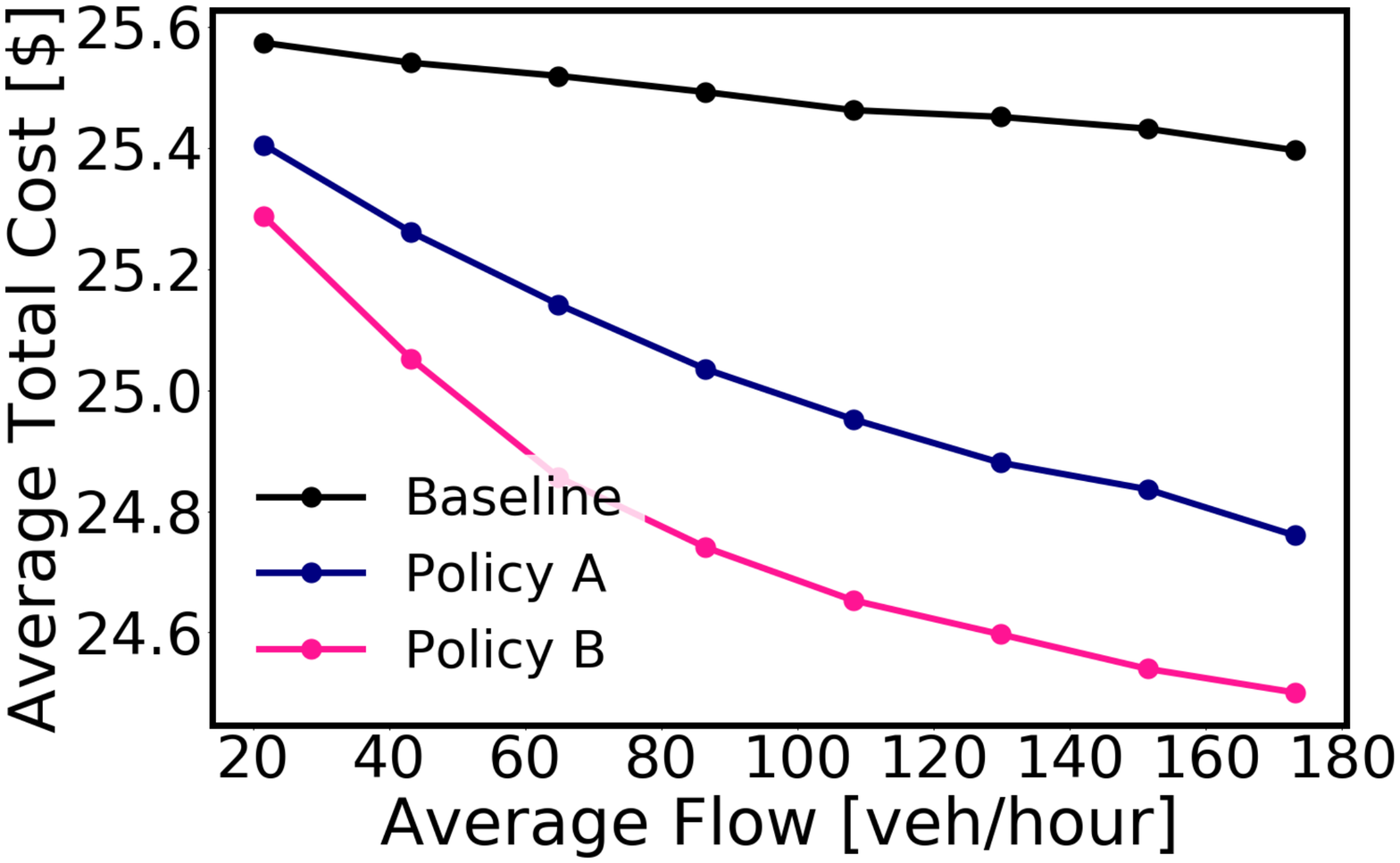}}
  \subfigure[Average fuel consumption]
  {\label{fig:average_fuel_consumption}
  \includegraphics[width=0.32\textwidth,trim=30 90 110 120,clip]{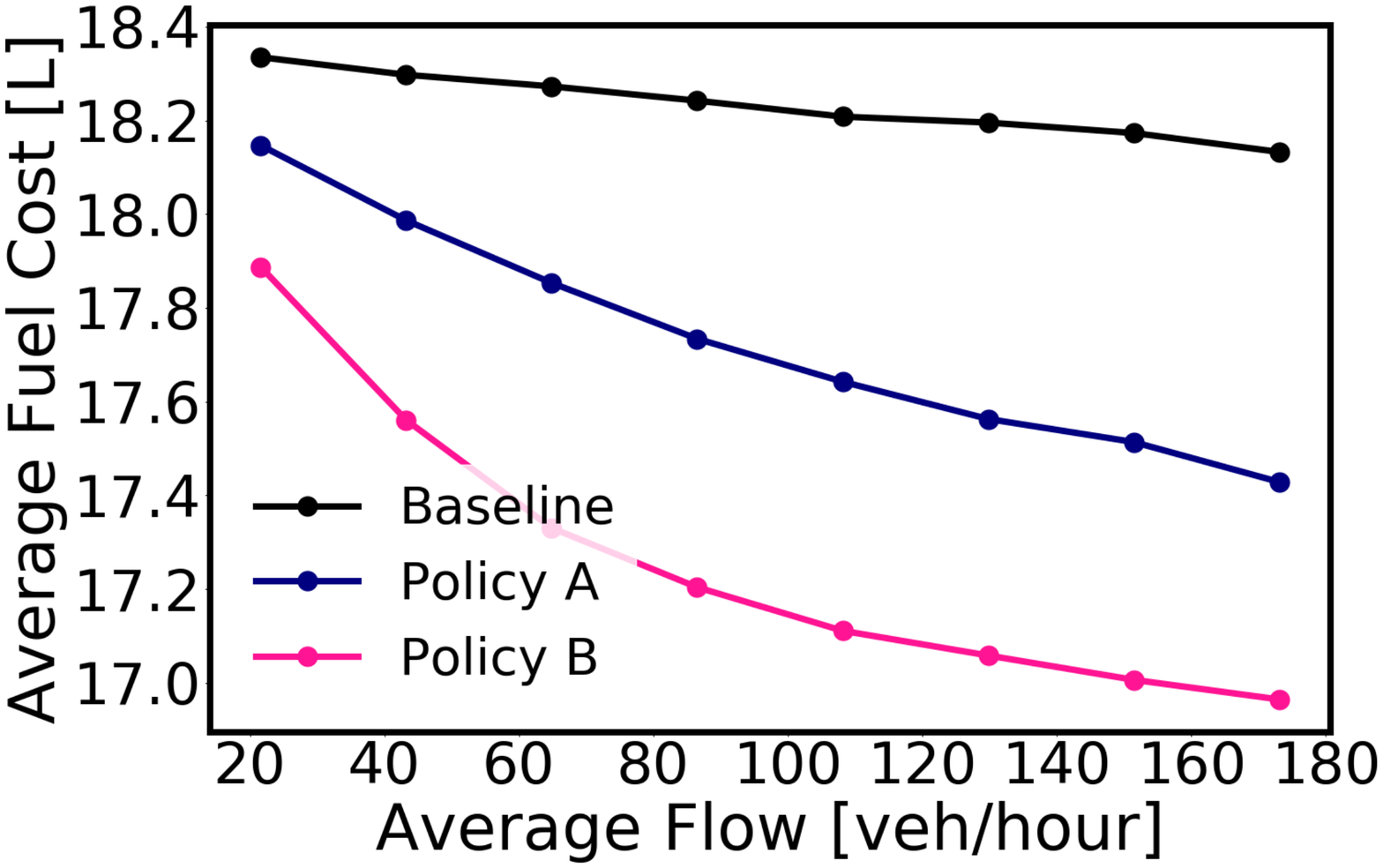}}
  \subfigure[Average travel time]
  {\label{fig:average_travel_time}
  \includegraphics[width=0.32\textwidth,trim=30 90 110 120,clip]{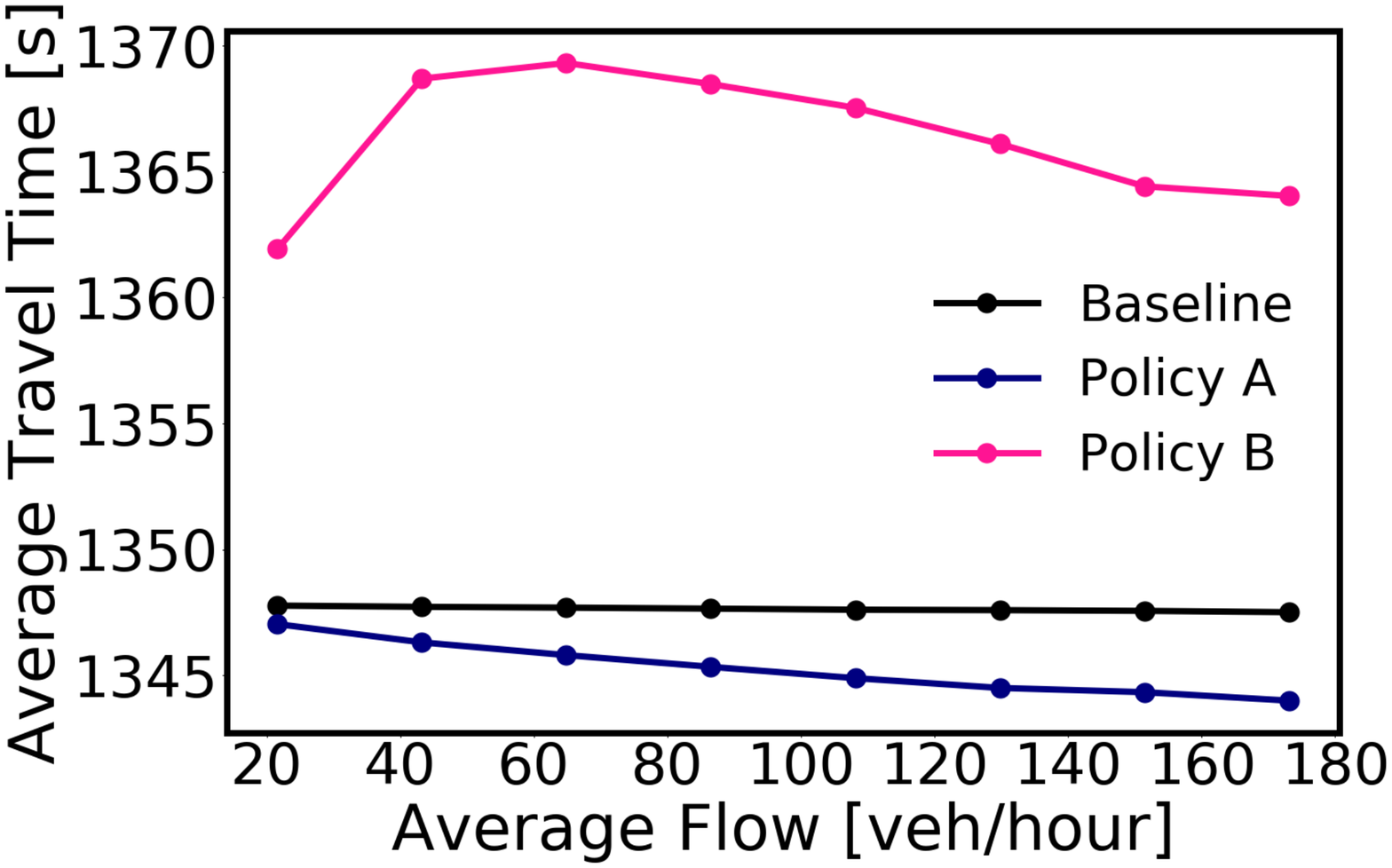}}
  
  \caption{Performance comparison among various policies.} 
\end{figure}

In Figure \ref{fig:average_total_cost}, Policy B yields $ \$ 0.9$ average cost reduction from the baseline when the average flow is $173 \ veh/hour$, i.e., we can save $\$ 3736.8 $ at one junction after one-day coordination. Policy A can reduce $ \$ 0.6$ average cost, and the total cost reduction is only $\$2491.2$. Hence, policy B is more effective in reducing the absolute total cost. Figure \ref{fig:average_fuel_consumption} shows the average fuel consumption under different flows, and Figure \ref{fig:average_travel_time} shows the average travel time with flows. The two figures show that Policy B has the best performance by reducing the fuel consumption. The average travel time in Policy B is even larger than the baseline. Policy A reduces both fuel consumption and travel time, however, the total cost is more than Policy B.

\subsection{Sensitivity analysis}
\label{sub_sensitivity}
In this subsection, we study the sensitivity of the average cost with key parameters: discount factor $\gamma$ and cursing distance $D_2$. We have conducted the sensitivity analysis of $\gamma$ when $D_2 = 30 \ km$ and $D_2 = 70 \ km$. The average vehicle flow is $45 \ veh/hour$. The results in Figure \ref{fig:discount_factor_30} indicate that the average total cost would decrease as we increase $\gamma$, which denotes the extent of considering long-term rewards, when the cruising distance is $30 \ km$.  However, when we increase the cruising the distance to $70 \ km$ (Figure \ref{fig:discount_factor_70}), the curve would firstly decrease and then increase with $\gamma$. When we increase the cruising distance, the policy enables more deceleration to form longer platoons under large $\gamma$ due to fuel savings in $D_2$. When the following vehicle receives the non-merging signal, lower speed would increase the travel time, and thus increase the total cost. Analysis of the vehicle deceleration with respect to $D_2$ is shown in Section~\ref{sub_interpretation}.

\begin{figure}[htbp]
  \centering
  \subfigure[$D_2 = 30 \ km$]
  {\label{fig:discount_factor_30}
  \includegraphics[width=0.42\textwidth,trim=40 90 90 120,clip]{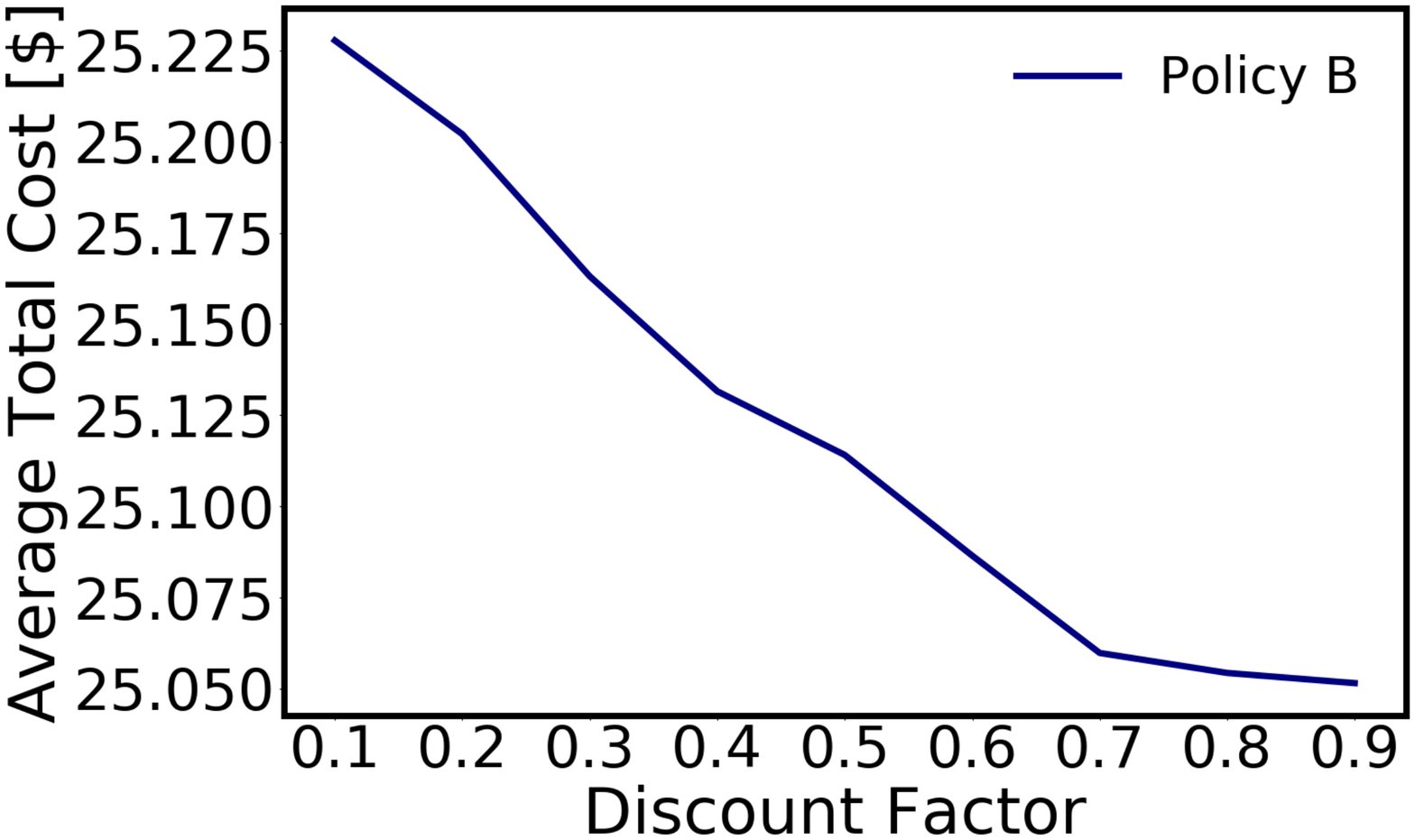}
  }
  \subfigure[$D_2 = 70 \ km$]
  {\label{fig:discount_factor_70}
  \includegraphics[width=0.42\textwidth,trim=40 90 90 120,clip]{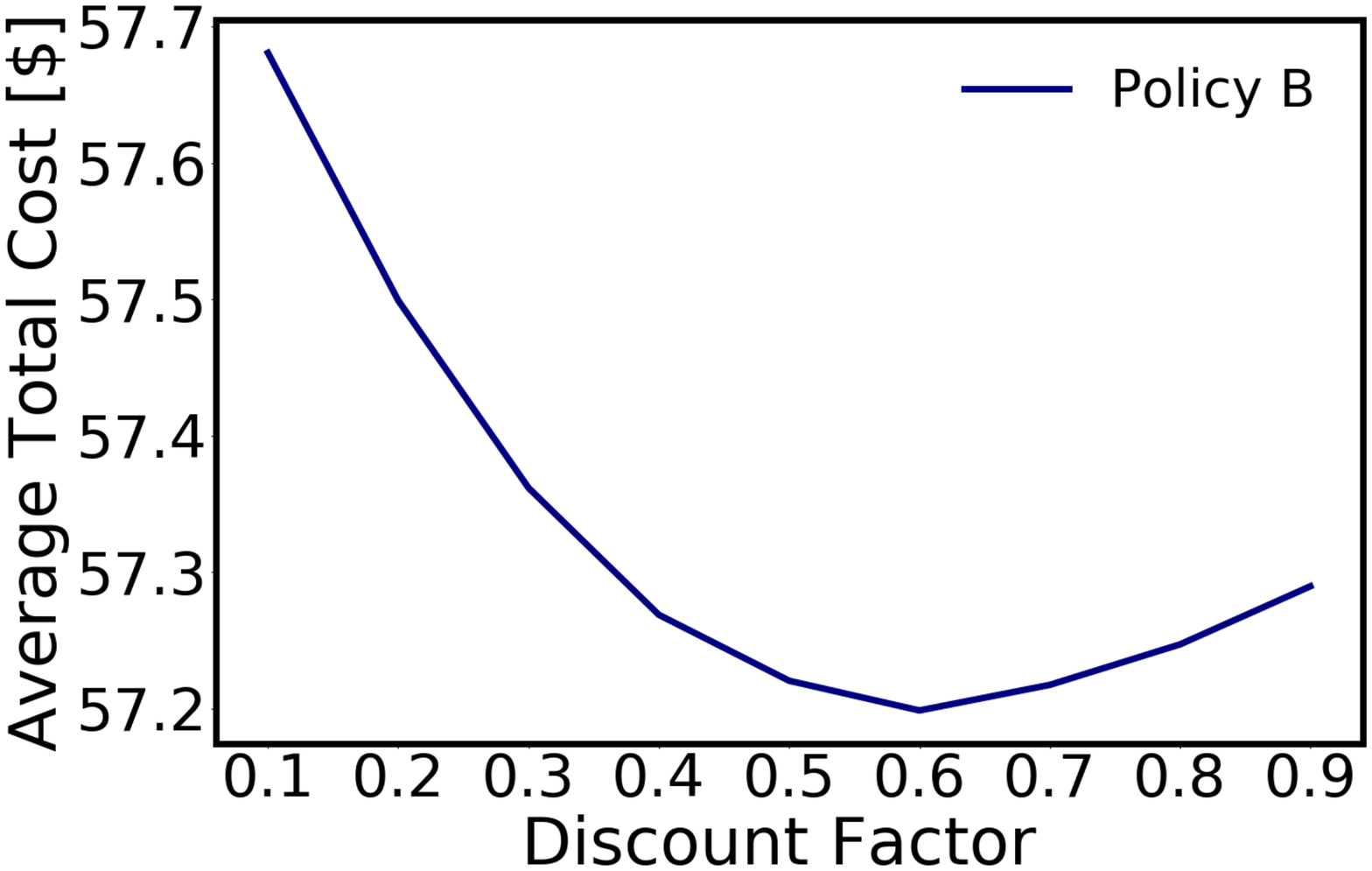}
  }
  \caption{Average cost with the discount factor $\gamma$.} \label{fig:discount_factor}
\end{figure}

Then we investigate the effect of the cruising distance $D_2$. The evaluation metric is the average cost per kilometer, $AC' = \frac{TC}{N_t * (D_1 + D_2)}$. Figure \ref{fig:cruising_distance_gamma_0.6} shows the average cost $AC'$ with $D_2$ under $\gamma = 0.6$. Average cost would decrease as we increase the cruising distance due to the fuel reduction in $D_2$. However, the results in Figure \ref{fig:cruising_distance_gamma_0.9} show that the average cost per kilometer curve firstly decreases, and then does not change too much after $60 \ km$. The initial decrease of $AC'$ results from the benefits of fuel savings in $D_2$. The marginal improvement of extending the cruising distance is negligible after $60 \ km$ with a large $\gamma$. Figure \ref{fig:discount_factor} and Figure \ref{fig:cruising_distance} indicate that large discount factor and long cruising distance together can increase the average total cost.

\begin{figure}[htbp]
  \centering
  \subfigure[$ \gamma = 0.6$]
  {\label{fig:cruising_distance_gamma_0.6}
  \includegraphics[width=0.42\textwidth,trim=60 90 90 130,clip]{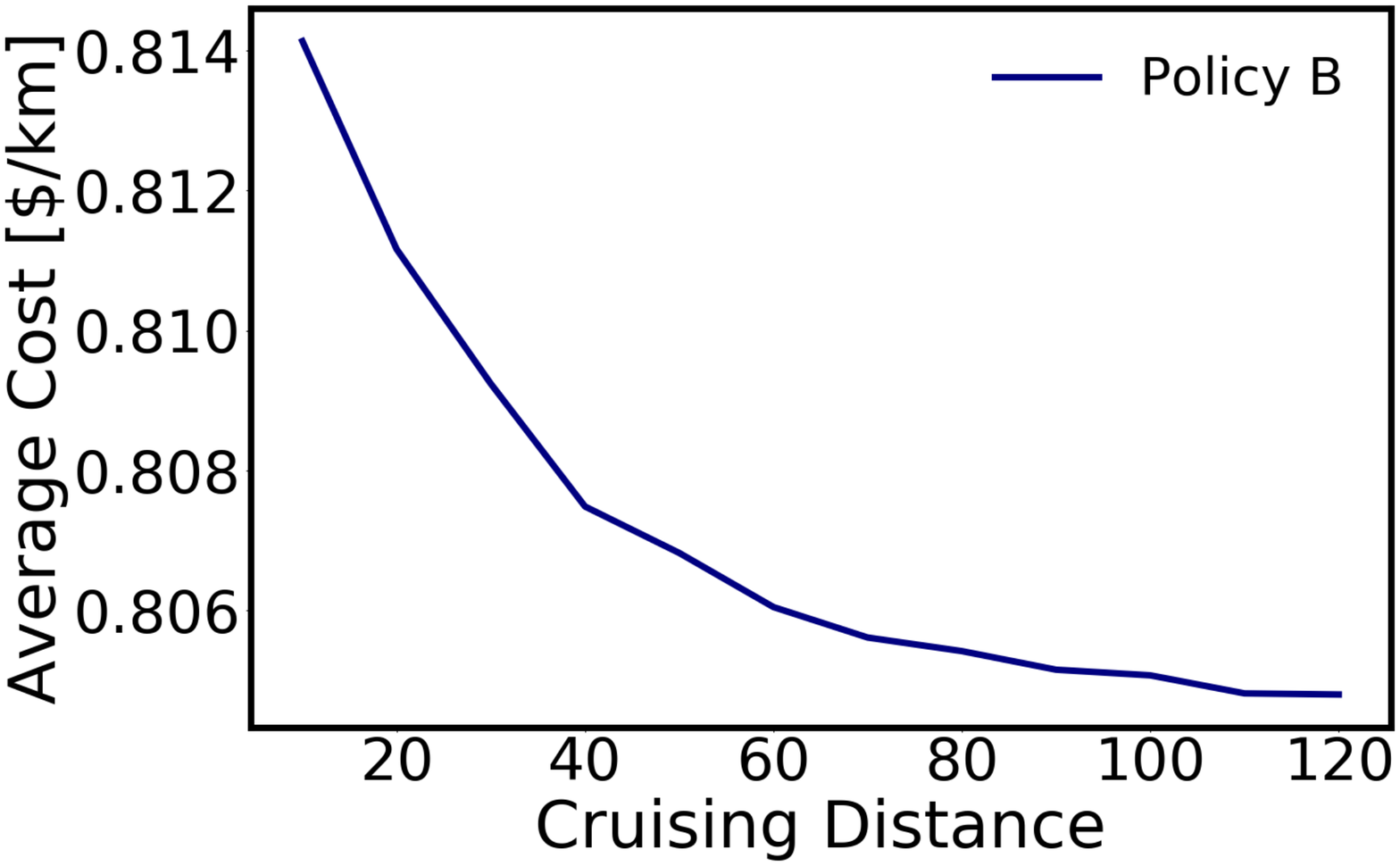}
  }
  \subfigure[$ \gamma = 0.9$]
  {\label{fig:cruising_distance_gamma_0.9}
  \includegraphics[width=0.42\textwidth,trim=60 90 90 130,clip]{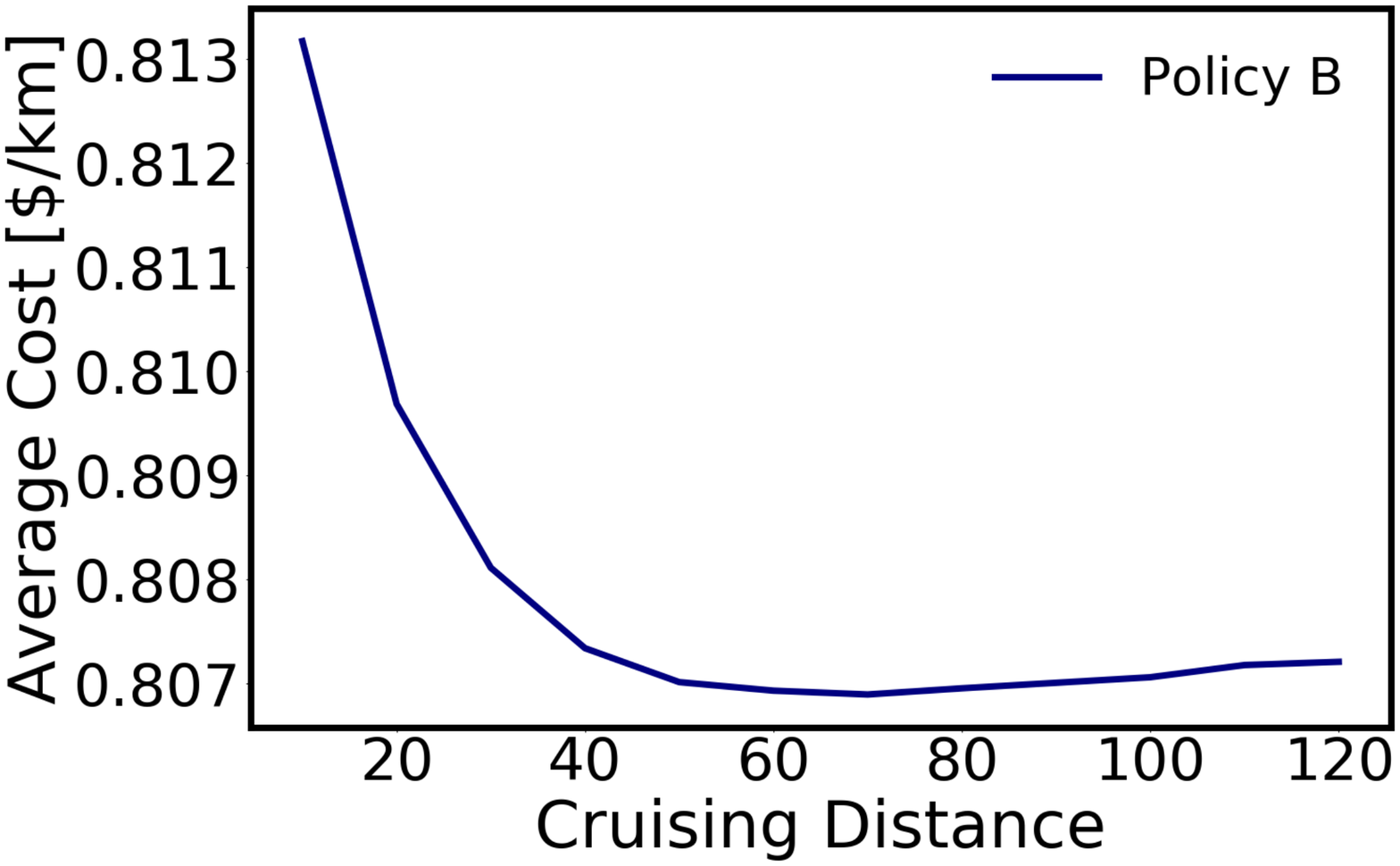}
  }
  \caption{Average cost per kilometer with the cruising distance $D_2$.} \label{fig:cruising_distance}
\end{figure}

\subsection{Interpretation of Real-time Strategy}
\label{sub_interpretation}
In this subsection, we present the insights of the Real-time Strategy (RTS). Figure \ref{fig:threshold_C_with_flow} shows the dynamic threshold and C value with respect to the estimated flow $\hat\lambda_k$. The analysis period is from 9:10 A.M. to 9:57 A.M., and the number of vehicle arrival is 50.
When we increase $D_2$ from $30 \ km$ to $70 \ km$, the threshold would not change too much, while the C values in Figure \ref{fig:C_with_flow_70} are much smaller than those in Figure \ref{fig:C_with_flow_30}. When $D_2 = 70 \ km$, the leading vehicle would have smaller speed due to more deceleration, which would reduce the platoon average speed.

\begin{figure}[htb]
  \centering
  \subfigure[Threshold when $D_2 = 30 \ km$.]
  {\label{fig:threshold_with_flow_30}
  \includegraphics[width=0.47\columnwidth,trim=20 110 90 140,clip]{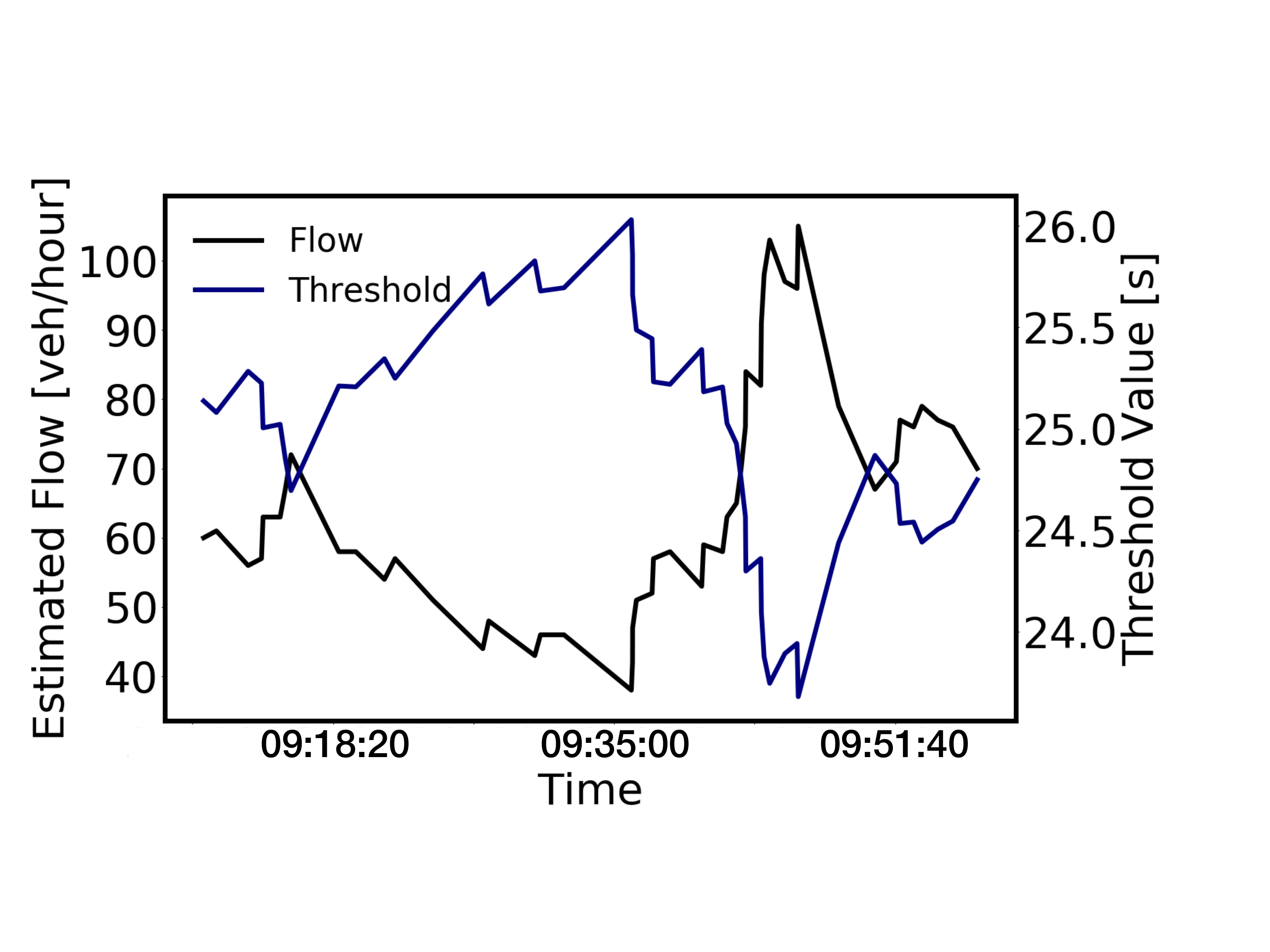}
  }
  \subfigure[C value when $D_2 = 30 \ km$.]
  {\label{fig:C_with_flow_30}
  \includegraphics[width=0.47\columnwidth,trim=20 110 90 140,clip]{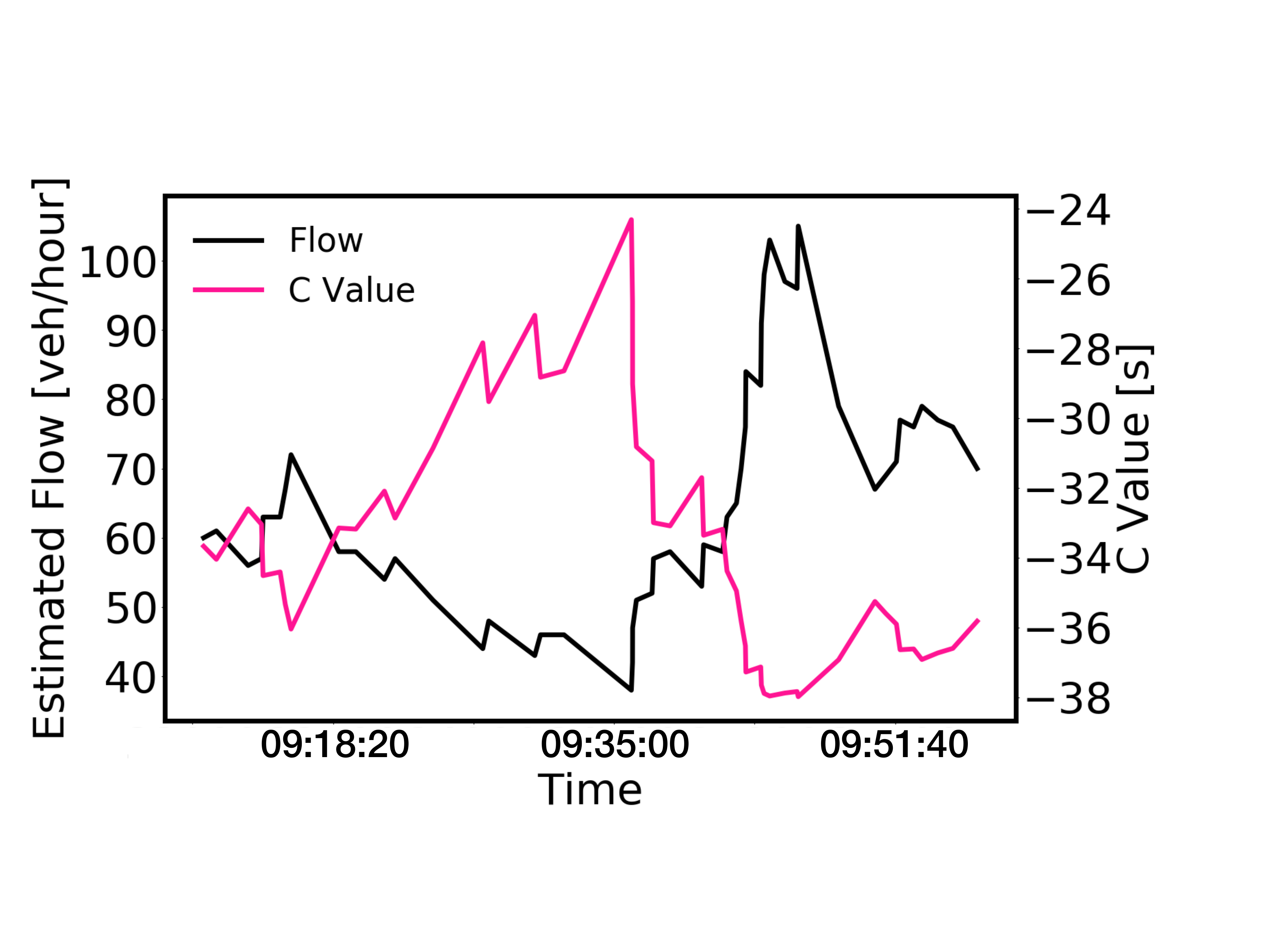}
  }
  \subfigure[Threshold when $D_2 = 70 \ km$.]
  {\label{fig:threshold_with_flow_70}
  \includegraphics[width=0.47\columnwidth,trim=20 110 90 140,clip]{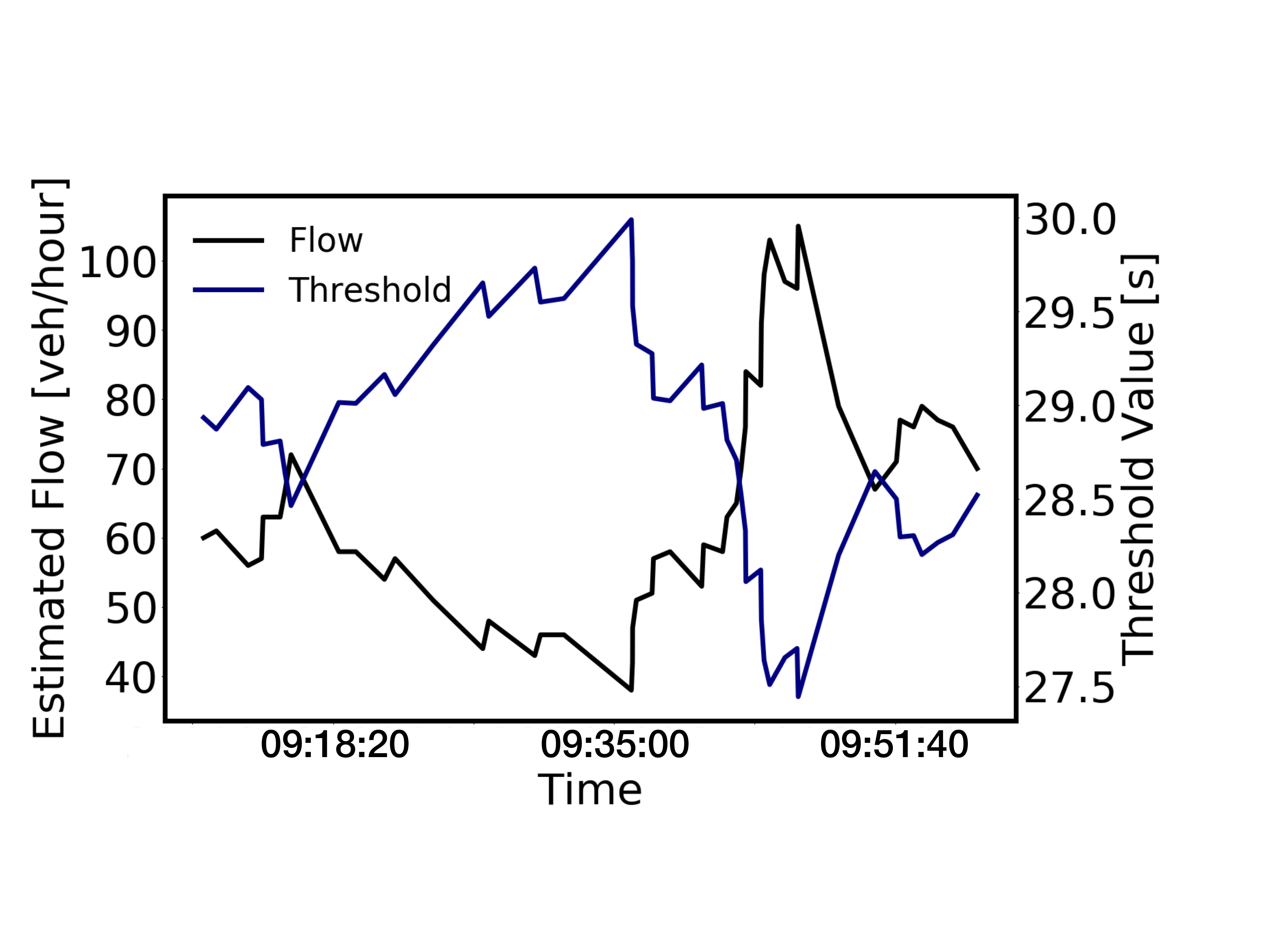}
  }
  \subfigure[C value when $D_2 = 70 \ km$.]
  {\label{fig:C_with_flow_70}
  \includegraphics[width=0.47\columnwidth,trim=20 110 90 140,clip]{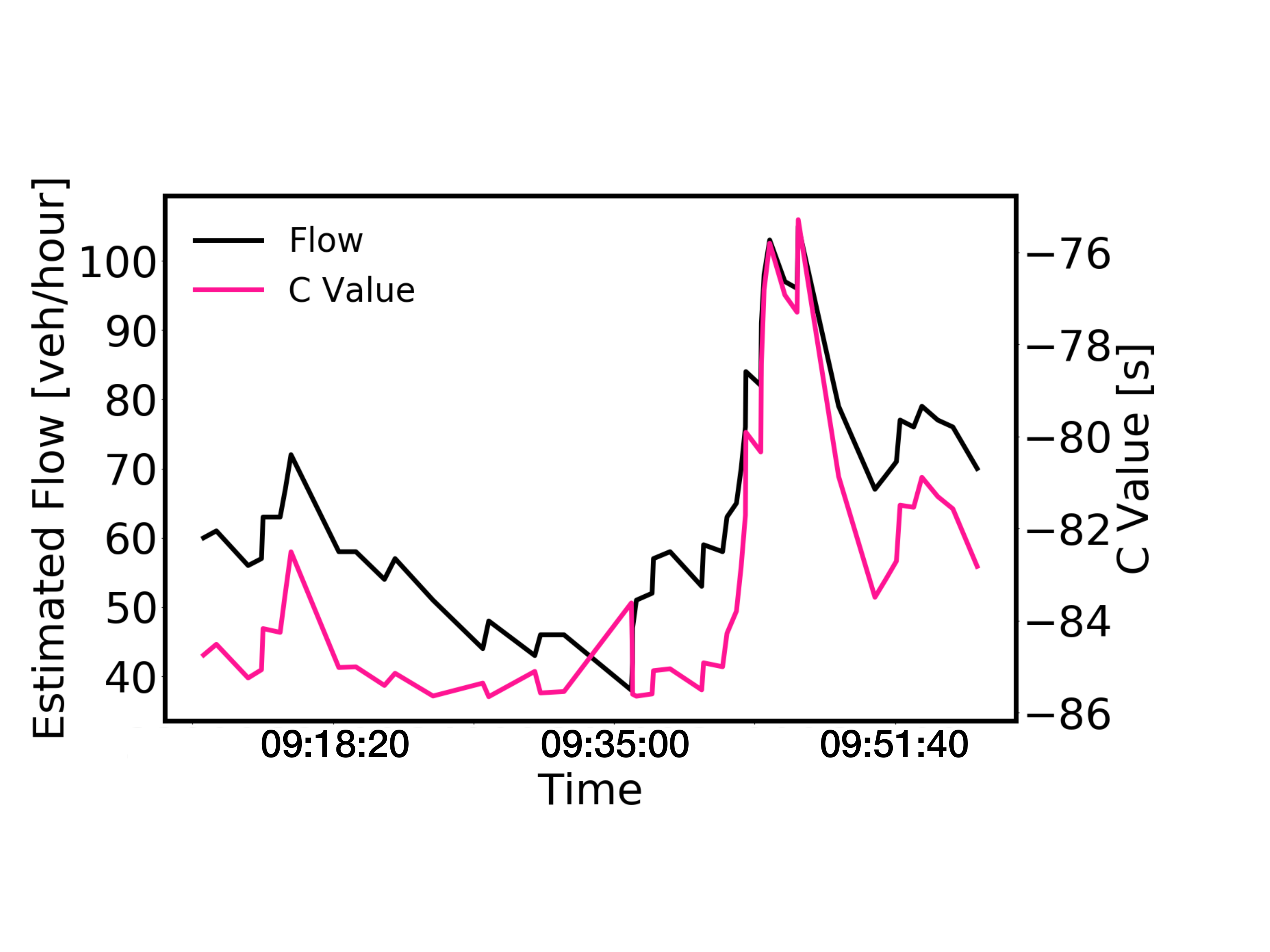}
  }
  
  \caption{Threshold and C with the estimated flow $\hat\lambda_k$.} \label{fig:threshold_C_with_flow}
\end{figure}

In Figure \ref{fig:threshold_with_flow_30} and Figure \ref{fig:threshold_with_flow_70}, the threshold changes toward the opposite direction of the flow, which indicates that lower flow can increase the threshold and thus enable more platooning. The C values in Figure
\ref{fig:C_with_flow_30} change toward the opposite direction of the flow, however, the C values in Figure \ref{fig:C_with_flow_70} change toward the same direction. The gap between threshold and C, i.e., threshold minus the C, would increase as we have longer $D_2$.

We then select 5 vehicles during the analysis period to show their longitudinal maneuvers by the space-time diagram (Figure \ref{fig:space-time diagram}) in the coordinating zone $D_1$. In Figure \ref{fig:space_time_30}, these 5 vehicles form 2 platoons with higher speeds under $D_2 = 30 \ km$. In contrast, vehicles form 1 platoon with lower speeds in Figure \ref{fig:space_time_70} under $D_2 = 70 \ km$, which indicates that higher cruising distance can facilitate the platooning process and reduce the average speed in $D_1$, thus increase the total cost resulting from more travel time. The average cost of extending $D_2$ would firstly decrease and then increase slightly in Figure \ref{fig:cruising_distance_gamma_0.9}.

\begin{figure}[hbt]
  \centering
  \subfigure[$D_2 = 30 \ km$]
  {\label{fig:space_time_30}
  \includegraphics[width=0.45\textwidth, trim=120 150 140 130,clip]{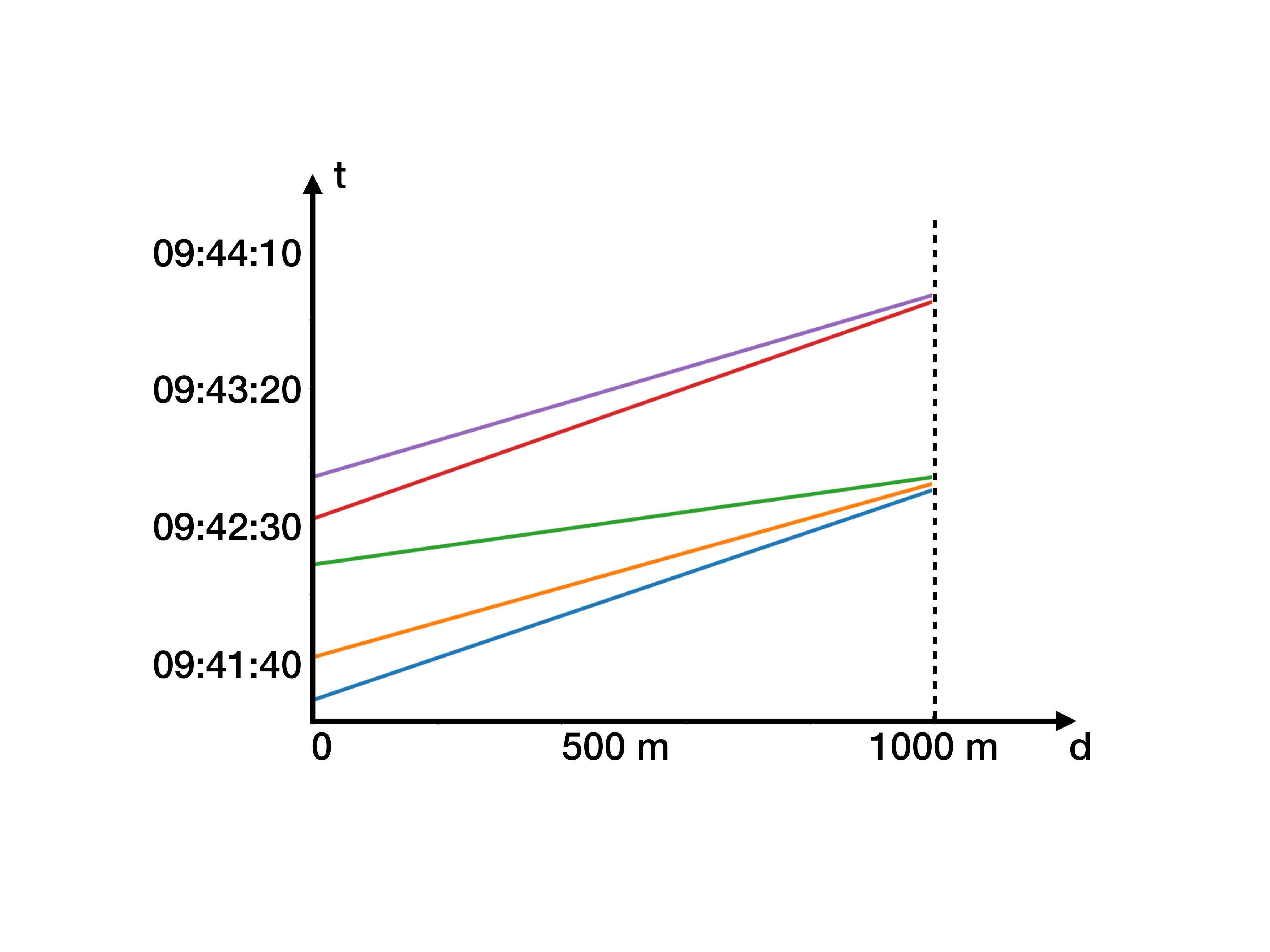}
  }
  \subfigure[$D_2 = 70 \ km$]
  {\label{fig:space_time_70}
  \includegraphics[width=0.45\textwidth, trim=120 150 140 130,clip]{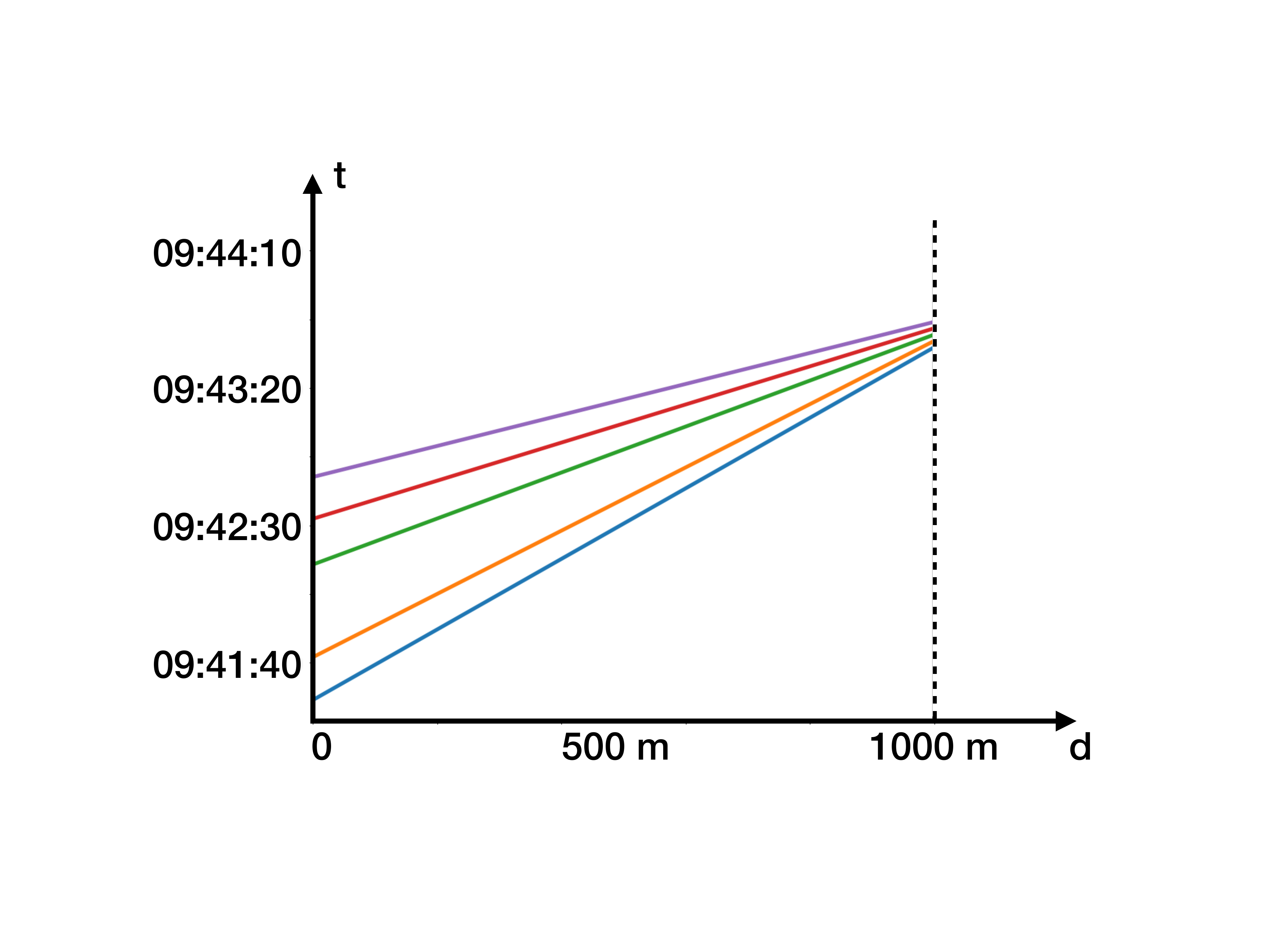}
  }
  \caption{Space-time diagram in the coordinating zone} \label{fig:space-time diagram}
\end{figure}
\section{Concluding Remarks}
\label{sec_conclude}

In this paper, we consider the coordinated platooning problem using a Markov decision process formulation.
The formulation is general in the sense that the arrival process can be a rather general class of renewal processes and that the cost function is generic.
By studying the Bellman optimality condition for the MDP, we show that the optimal coordination strategy is threshold based.
Using this structural result, we develop a recursive approximation algorithm to compute the optimal strategy, which is significantly faster than the generic value iteration algorithm.
Furthermore, we show that for Poisson arrival processes, the optimal strategy can be directly computed by solving a system of integral equations.
We also validate our results in simulation with Real-time Strategy using real traffic data.

This work can be extended in several directions.
First, the analysis for a single junction is the basis for the analysis of centralized or distributed coordination at networks of junctions (see e.g. \cite{xiong2019evaluation} for a preliminary simulation-based analysis).
Second, the interaction between CAV platoons and the background traffic still needs to be appropriately modeled and addressed.
Third, social-economical mechanisms that incent CAVs to cooperate still needs to be designed and validated, since the benefit of platooning is not evenly distributed over all vehicles.

\bibliographystyle{IEEEtran}
\bibliography{optimization}

\end{document}